\documentclass[pdflatex,sn-mathphys-num]{sn-jnl}

\usepackage{graphicx}%
\usepackage{multirow}%
\usepackage{amsmath,amssymb,amsfonts}%
\usepackage{amsthm}%
\usepackage{mathrsfs}%
\usepackage[title]{appendix}%
\usepackage{xcolor}%
\usepackage{textcomp}%
\usepackage{manyfoot}%
\usepackage{booktabs}%
\usepackage{algorithm}%
\usepackage{algorithmicx}%
\usepackage{algpseudocode}%
\usepackage{listings}%
\usepackage{graphicx, mathtools, bm, bbm, ulem, braket, tikz, url, hyperref, cleveref, enumitem} 
\usetikzlibrary{calc}
\usepackage{subcaption}
\usepackage{caption, float}
\usetikzlibrary{positioning}

\theoremstyle{thmstyleone}%
\newtheorem{theorem}{Theorem}[section]
\newtheorem{proposition}[theorem]{Proposition}%

\theoremstyle{thmstyletwo}%
\newtheorem{example}{Example}%
\newtheorem{remark}{Remark}%

\theoremstyle{thmstylethree}%
\newtheorem{definition}{Definition}%

\newcommand{\Tr}{\mathrm{Tr}}
\newcommand{\Ad}{\operatorname{Ad}}
\numberwithin{equation}{section}

\begin{document}

\title[Pauli Supported Invariants for Periodic Graphs-Derived Hamiltonians]{Pauli Supported Invariants for Periodic Graphs-Derived Hamiltonians}

\author*[1,2]{\fnm{Sarah} \sur{Chehade}}\email{sschehad@central.uh.edu}

\author[3]{\fnm{Andrew} \sur{Vlasic}}\email{avlasic@deloitte.com}

\author[4]{\fnm{Rebekah} \sur{Herrman}}\email{rherrma2@utk.edu}

\author[1,2]{\fnm{Rick} \sur{Mukherjee}}\email{rick-mukherjee@utc.edu}

\affil[1]{\orgdiv{Quantum Center}, \orgname{University of Tennessee, Chattanooga}, \city{Chattanooga}, \postcode{37403}, \state{TN}, \country{USA}}

\affil[2]{\orgdiv{Department of Physics and Astronomy, University of Tennessee}, \city{Chattanooga}, \postcode{37403}, \state{TN}, \country{USA}}

\affil[3]{\orgname{Deloitte Consulting LLP}, \city{Tampa}, \postcode{33602}, \state{FL}, \country{USA}}

\affil[4]{\orgdiv{Industrial and Systems Engineering}, \orgname{University of Tennessee, Knoxville}, \postcode{37996}, \state{TN}, \country{USA}}


\abstract{ We introduce a graph invariant obtained from Pauli decompositions of Hamiltonians derived from local graph neighborhoods. Given a rooted h-hop neighborhood, we construct a local adjacency operator, embed it into a common Hilbert space dimension, and define its Pauli-support set as the collection of Pauli strings appearing with nonzero coefficients. For periodic graphs, we prove invariance under lattice-compatible graph equivalence and establish converse results under root-separation hypothesis. The framework naturally extends to Lie closures, commutants, and Cartan-type structures generated by the associated Pauli supports. We further discuss extensions to aperiodic graphs with finite local complexity and make connections to cut-and-project models of quasicrystals and notably, Penrose tilings. From a quantum information perspective, the resulting invariants provide a operator-theoretic description of graph-based Hamiltonians and offer a new mechanism for comparing graph-based quantum systems through Pauli-support data. }

\keywords{graph invariants, quantum information theory, periodic graphs, crystalline graphs, Hamiltonian representations, quasicrystals}

\maketitle

\section{Introduction}
    Understanding how local structure determines global behavior is a recurring theme across mathematics, physics, materials science, and quantum information~\cite{lovasz2012large, sachdev1999quantum, cirac2012goals, altman2021quantum}. In graph theory, this question appears most prominently in the study of graph isomorphism, which determines whether two graphs are structurally identical. While graph isomorphism is efficiently solvable for many important classes, such as bounded degree~\cite{luks1982isomorphism}, it remains a challenging problem in general and identifying computable invariants that reliably distinguish graphs is notoriously difficult~\cite{babai2016graph}. 

    This difficulty motivates the search for more structured settings in which local to global constructions become more tractable. Periodic graphs, also known as crystals or space graphs, provide such a setting. These graphs arise naturally in lattice systems, condensed matter models, and structured quantum circuits~\cite{sachdev1999quantum, sunada2012topological, childs2004spatial, gonzales2025efficient, herrman2022simplifying, atallah2026simulating}, and are characterized by a finite repeating unit together with a lattice action that generates the full infinite structure. The presence of this additional symmetry suggests that global structures may be recoverable from sufficiently rich local data. Recent work indicates that quantum computing may be useful for identifying if graphs with particular structure are isomorphic or not by leveraging the local structure \cite{gaitan2014graph, mills2019quantum, szegedy2019qaoa, douglas2008classical}. 

    A fundamental question is therefore whether periodic structures can be compared and classified using local invariants. While graph isomorphism provides a notion of equivalence for finite graphs, periodic graphs require a refinement that incorporates both combinatorial structure and lattice translations. In this setting, two periodic graphs should be considered equivalent only if there exists a graph isomorphism that preserves both the graph structure and the periodic translation structure induced by the lattice.  
    In this work, we introduce a new invariant for periodic graphs based on Pauli support. Given a graph, we associate to each local neighborhood a Hamiltonian encoding of its adjacency structure and expand this operator in the Pauli basis. The resulting support, which we call the root set, captures algebraic information about the local structure. Aggregating this data across the graph yields a computable invariant that is sensitive to both local connectivity and global periodic structures. 

    Our perspective is motivated by the central role of Pauli expansions and operator algebras in quantum information~\cite{vlasic2024quop}. From a quantum information perspective, our proposed invariant using Pauli-supports provides a natural language for describing graph-derived Hamiltonians \cite{hantzko2024tensorized, van2020circuit}. Such Hamiltonians naturally arise in quantum simulations~\cite{cody2012faster}, quantum walks~\cite{gonzales2025efficient, herrman2022simplifying, atallah2026simulating}, controllability analysis~\cite{gessner2014nonlinear}, and circuit synthesis~\cite{furrutter2024quantum}. Consequently, Pauli-support invariants provide a mechanism for comparing graph-based quantum systems through operator-theoretic data rather than purely combinatorial representations. 

    We show that under suitable conditions, root sets provide a graph invariant for local structures in periodic graphs. To validate this approach, we present a series of examples. 
    Beyond periodic systems, we explore the applicability of this framework to aperiodic structures, also known as quasicrystalline structures, such as Penrose tilings. Although these systems lack strict periodicity and fall outside the scope of our main theorems, our numerical experiments suggest that root set invariants remain sensitive to their local structure and may provide a promising tool for distinguishing aperiodic graphs. 


    The central theme of this work is that periodic symmetry reduces an infinite graph to finitely many local combinatorial types. We encode these local structures through Pauli-support data, yielding an operator-theoretic invariant that is preserved under lattice-compatible graph equivalence. Under a natural root-separation condition, this invariant also recovers the underlying local graph structure. Finally, we show that finite local complexity plays an analogous finiteness role in the aperiodic setting, allowing the framework to extend beyond periodic graphs. 
    
    The rest of the paper is structured as follows: In \Cref{section:prelim}, we review graph-theoretic and lattice-theoretic concepts required throughout the paper. In \Cref{section: periodic nbhds}, we establish the relationship between local neighborhoods and periodic graph structures, showing how sufficiently large neighborhoods encode complete motif information. In \Cref{section: local invariant}, we introduce the local adjacency operators, their Pauli decompositions, and the associated root-set invariants, and we establish their invariance under lattice-compatible graph equivalence. In \Cref{section:converse}, we develop partial converse results for periodic graphs through the notion of root separation and show how global graph structures can be recovered from local invariant data. In \Cref{section: Lie}, we evaluate the operator-algebraic structure generated by our invariant by looking at Lie closures, commutants, and Cartan decompositions. Finally in \Cref{section:aperiodic}, we extend this framework to aperiodic graphs with finite local complexity and discuss connections with cut-and-project quasicrystalline structures.

\section{Preliminaries}\label{section:prelim}
 Since the proposed framework lies at the intersection of graph theory, quantum information, and Lie theory, with applications to quantum computing and materials science, readers from different backgrounds may find certain topics more familiar than others. For completeness, we briefly review the graph-theoretic concepts required for the main results. Our discussion of Lie algebras is saved for \Cref{section: Lie} and focuses only on the structural properties used throughout the paper and is not intended as a comprehensive introduction to the subject.    

\subsection{Graph Theory}
A graph $G = (V(G),E(G))$ is a collection of vertices $v \in V(G)$ and edges $(u,v) \in E(G)$ for $u, v \in V(G)$. In this work, we consider only simple graphs, which have no multi-edges or self-loops. A vertex $u$ is said to be a \textit{neighbor} of vertex $v$ if $(u,v) \in E(G)$. The set of all neighbors of $v$, $\{ u : (u,v) \in E(G)\}$, is called the \textit{neighborhood} of $v$ and is denoted $N_G(v)$. The \textit{degree} of vertex $v$ is the cardinality of the neighborhood of $v$, $|N_G(v)|$. The \textit{degree sequence} of $G$ is an ordered list of the degree of each vertex of $G$ in non-increasing order. Note that when $G$ is clear, we shall drop the parenthetical and subscript $G$ in the above definitions. 

Two graphs $G$ and $H$ are said to be \textit{isomorphic} if there exists a bijection $f: V(G) \rightarrow V(H)$ such that $(u,v) \in E(G)$ if and only if $(f(u), f(v)) \in E(H)$.  Note that two graphs must have the same number of vertices and edges to be isomorphic, and the degree sequence of both $G$ and $H$ must be the same. However this is not a sufficient condition to prove isomorphism. 

\subsection{Crystal Graphs}
Crystal graphs are objects where the periodicity of the structure can be infinitely repeated. For instance, periodic graphs can be used as a tiling of a Euclidean plane. The following three definitions, while mathematical, were derived as a means to fully describe this structure. Given this structure, a proposition about the useful automorphism is then given.  

The following definition is that of a \textit{lattice}, which gives a general framework for periodic graphs.
\begin{definition}\label{def:lattice}
    Let $d\in\mathbb{N}$. A lattice in $\mathbb{R}^d$ is a discrete subgroup 
    \begin{equation}
        \Lambda \coloneqq \left\{ \sum\limits_{j=1}^d n_ja_j : n_j\in\mathbb{Z} \right\},
    \end{equation}
    where $a_1, \dots, a_d\in\mathbb{R}^d$ are linearly independent.
\end{definition}

This lattice structure, while essential, is too general for a crystal. Particularly, that of invariance under translations, describing the long-range periodic order. The next definition makes this criterion rigorous. 
\begin{definition}
    A periodic set in $\mathbb{R}^d$ is a discrete set of points $X\subset\mathbb{R}^d$ that is invariant under translation by a lattice $\Lambda$, i.e., 
    \begin{equation}
        X + \lambda = X \quad\text{for all}\quad \lambda\in\Lambda.
    \end{equation}
\end{definition}

\begin{definition}
    Let $\Lambda$ be a lattice. A graph $G=(V,E)$ embedded in $\mathbb{R}^d$ is called a $\Lambda$-space graph if 
    \begin{enumerate}
        \item $V\subset\mathbb{R}^d$ is discrete,
        \item $V + \Lambda = V$ (vertex set is invariant under translation), 
        \item $G$ is invariant under lattice translations. That is
        \begin{equation}
            (v,w)\in E \iff (v+\lambda, w+\lambda)\in E, \quad\text{for all}\quad \lambda\in\Lambda,
        \end{equation}
        \item The quotient graph $G/\Lambda$ has finitely many vertices. 
    \end{enumerate} 
\end{definition}

Given a periodic space graph, we introduce a computationally tractable notion of automorphism. Recall an automorphism is a map from an object to itself, where the map is a bijection. The following proposition provides the corresponding characterization.
\begin{proposition}\label{prop: graph auto}
    Let $G=(V,E)$ be a $\Lambda$-space graph. Then for every $\lambda\in\Lambda$, the translation map
    \begin{equation}
        \tau_\lambda : V\to V, \qquad \tau_\lambda(v) \coloneqq v+\lambda
    \end{equation}
    is a graph automorphism. 
\end{proposition}

\begin{proof}
        Fix $\lambda\in\Lambda$. Since $G$ is a space graph with respect to the lattice $\Lambda$, it follows that the vertex set is invariant under lattice translations. Hence
        \begin{equation}
            v\in V \implies v+\lambda\in V + \Lambda = V.
        \end{equation}
        So $\tau_\lambda: V\to V$ holds. Moreover, the inverse map is precisely $\tau_{-\lambda}$ because 
        \begin{equation}
            \tau_{-\lambda}\tau_\lambda(v) = \tau_{-\lambda}(v+\lambda) = v,
        \end{equation}
        and similarly for $\tau_\lambda\tau_{-\lambda}$. Hence $\tau_\lambda$ is a bijection. Next, by translation invariance of edges in a space graph, we have that 
        \begin{equation}
            (u,v)\in E \iff (\tau_\lambda(u),\tau_\lambda(v))\in E.
        \end{equation}
        This completes the proof. 
\end{proof}

The graph automorphism in the proposition establishes a basis to conduct analysis on a periodic graph by restricting the entire graph to a finite number of computationally tangible sets, particularly, a `generator' set that extends to the entire graph from shifts.

\begin{definition}
    Let $(G,\Lambda)$ be a periodic graph. For a vertex $v\in V(G)$, the orbit of $v$ under $\Lambda$ is
    \begin{equation}
        \mathcal{O}(v)
        \coloneqq
        \{ v + \lambda : \lambda \in \Lambda \}.
    \end{equation}
\end{definition}

The lattice $\Lambda$ partitions $V(G)$ into disjoint orbits under translation. Since $G/\Lambda$ has finitely many vertices, there are finitely many such orbits. A unit cell $V_Y$ contains exactly one representative from each orbit.

\begin{definition}
    A fundamental domain for a lattice $\Lambda$ is a subset $Y\subset\mathbb{R}^d$ where $$Y = \left\{ \sum_{i=1}^d t_i a_i : \forall i \ 0 \leq t_i < 1 \right\}$$ 
    such that 
    \begin{enumerate}
        \item $\mathbb{R}^d = \bigcup\limits_{\lambda\in\Lambda}(Y + \lambda)$,
        \item $(Y + \lambda_1)\cap (Y + \lambda_2) = \emptyset,$ for all $\lambda_1\neq\lambda_2$. 
    \end{enumerate}
    Intuitively, the fundamental domain is the one tile that generates all of the space under translation. 
\end{definition}

The fundamental domain provides the geometric building block of the periodic structure. For graph-theoretic and computational purposes, however, it is more convenient to work with the finite graph contained within a fundamental domain rather than the geometric region itself. This motivates the following notions of a \emph{motif} and a \emph{unit cell}. Note that the unit cell may have a different name in various disciplines. For instance, in physics, it has the name \emph{primitive cell}.
The motif is the finite graph corresponding to a unit cell, and the entire infinite space graph is recovered by translating the motif by all elements of the lattice. 

\begin{definition}
      \begin{enumerate}
        \item A unit cell of a space graph $G$ is a fundamental domain $Y$ restricted to the vertex set $V_Y\coloneqq V\cap Y$.
        \item A graph restricted to the unit cell is an induced finite graph we call the motif. This graph notationally is given by $G[Y]\coloneqq G[V_Y]$.
    \end{enumerate}  
\end{definition}    

To make these concepts more tangible, we walk through two different examples. Example \ref{ex:fundamental_domain} starts with a square lattice and argues why the proposed subset of $\mathbb{R}^2$ is the fundamental domain. Example \ref{ex:honeycomb} walks through how to computationally create a honeycomb lattice.

\begin{example}\label{ex:fundamental_domain}
    Let $V = \mathbb{Z}^2\subset\mathbb{R}^2$ be a square lattice defined using the relations
    \begin{align*}
        (x,y) \sim & (x\pm 1,y) \\
        (x,y) \sim & (x,y\pm 1).
    \end{align*}
    Then the lattice $\Lambda = \mathbb{Z}(1,0) \oplus \mathbb{Z}(0,1) = 
    \mathbb{Z}^2$ and the fundamental domain is $Y = [0,1)\times [0,1)$. Thus the unit cell is then $V_Y = \mathbb{Z}^2\cap Y$. Note that in this unit cell there is only one vertex, namely, $(0,0)$.
    
     This claim is not entirely obvious. One may observe that $\Lambda$ is a natural lattice as the shifts are one unit. First, we show $Y = [0,1)\times [0,1)$ is a fundamental domain. Let $(x,y)\in\mathbb{R}^2$. Define $m = \lfloor x\rfloor \text{ and } n = \lfloor y \rfloor$, where $\lfloor\cdot\rfloor$ denotes the floor function. Then we have that $x-m \text{ and } y-n \in [0,1)$, and hence $(x-m,y-n)\in Y.$ Since $(m,n)\in\Lambda$, it follows that $(x,y) = (x-m,y-n) + (m,n).$ Thus every point of $\mathbb{R}^2$ lies in some translate $Y+\lambda$ with $\lambda\in\Lambda$, proving $\mathbb{R}^2=\bigcup_{\lambda\in\Lambda}(Y+\lambda).$ 
     
     Next, we show that if $\lambda_1\neq \lambda_2$, then $(Y+\lambda_1)\cap (Y+\lambda_2)=\emptyset.$ For the sake of contradiction, suppose there exists a point $p\in (Y+\lambda_1)\cap (Y+\lambda_2).$ Then there exist $y_1,y_2\in Y$ such that $p=y_1+\lambda_1=y_2+\lambda_2.$ Rearranging gives $y_1-y_2=\lambda_2-\lambda_1\in\Lambda=\mathbb{Z}^2.$ Write $y_1=(a_1,b_1) \text{ and } y_2=(a_2,b_2),$ where $a_1-a_2\in(-1,1)$ and $ b_1-b_2\in(-1,1).$
        Since $y_1-y_2\in\mathbb{Z}^2$, both coordinates must be integers. But the only integer in $(-1,1)$ is $0$. Hence $a_1-a_2=0 \text{ and } b_1-b_2=0,$ so then $y_1=y_2$. From this, it follows that $\lambda_1=\lambda_2$, which contradicts the assumption. 
        
        This shows $Y=[0,1)\times[0,1)$ is a fundamental domain for $\Lambda$.
\end{example}

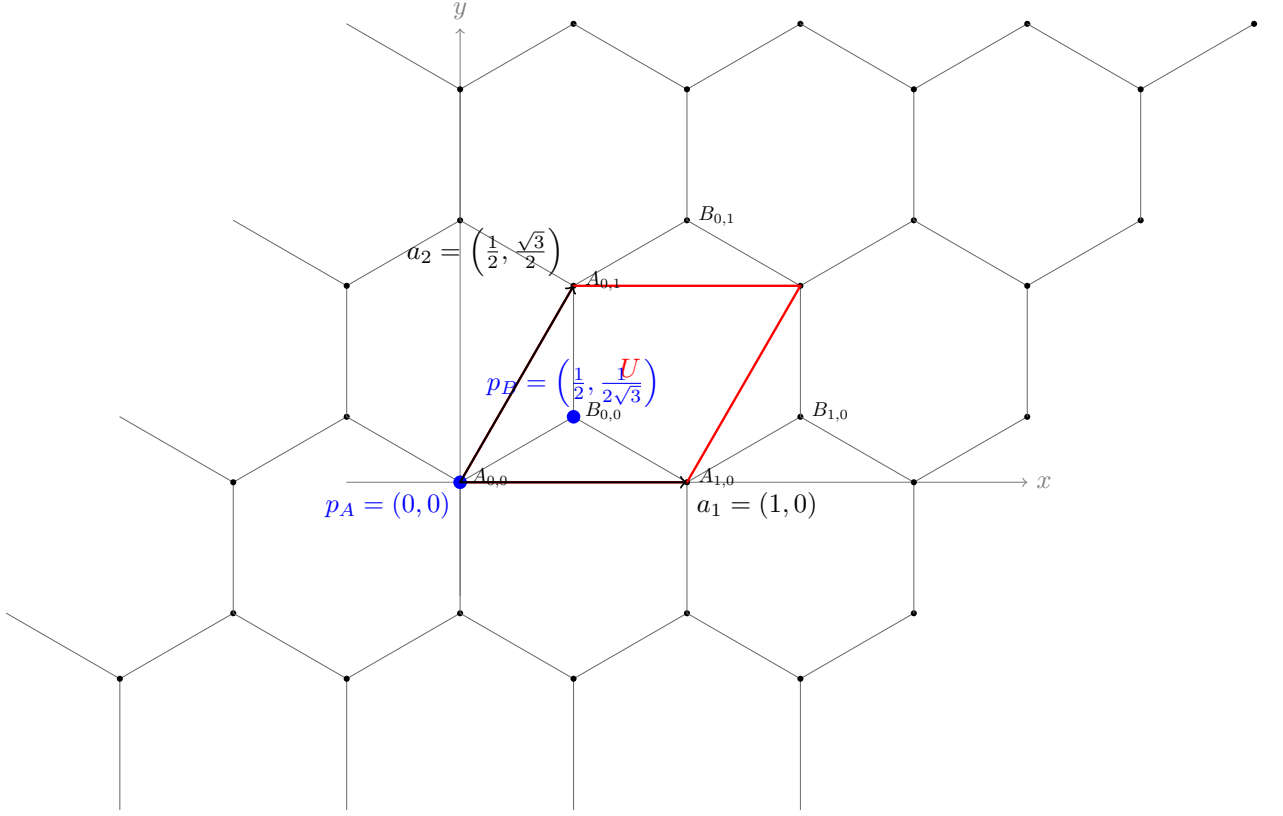
\begin{figure}
    \centering 
    \begin{tikzpicture}[scale=3, line cap=round, line join=round]    
        \draw[->, gray, thin] (-0.5,0) -- (2.5,0) node[right] {$x$};
        \draw[->, gray, thin] (0,-0.5) -- (0,2.0) node[above] {$y$};
        
        \def\s{0.866025403784}   
        \def\h{0.288675134595}   
        
        \coordinate (a1) at (1,0);                 
        \coordinate (a2) at (0.5,\s);              
        \coordinate (pA) at (0,0);                 
        \coordinate (pB) at (0.5,\h);              
        \foreach \m in {-1,0,1,2} {
          \foreach \n in {-1,0,1,2} {
        
            \coordinate (A)  at ($ (pA) + \m*(a1) + \n*(a2) $);
            \coordinate (B0) at ($ (pB) + \m*(a1) + \n*(a2) $);
            \coordinate (B1) at ($ (pB) + {(\m-1)}*(a1) + \n*(a2) $);
            \coordinate (B2) at ($ (pB) + \m*(a1) + {(\n-1)}*(a2) $);
        
            \draw[black!60, line width=0.35pt] (A) -- (B0);
            \draw[black!60, line width=0.35pt] (A) -- (B1);
            \draw[black!60, line width=0.35pt] (A) -- (B2);
        
            \fill[black] (A)  circle (0.38pt);
            \fill[black] (B0) circle (0.38pt);
          }
        }
        
        \coordinate (O)   at (0,0);
        \coordinate (U1)  at ($ (O) + (a1) $);
        \coordinate (U2)  at ($ (O) + (a2) $);
        \coordinate (U12) at ($ (O) + (a1) + (a2) $);
        
        \draw[red, line width=0.9pt] (O) -- (U1) -- (U12) -- (U2) -- cycle;
        \node[red] at ($ (O)!0.5!(U12) + (0,0.07) $) {$U$};
        
        \fill[blue] (pA) circle (0.85pt);
        \fill[blue] (pB) circle (0.85pt);
        \node[blue, anchor=north east] at (pA) {$p_A=(0,0)$};
        \node[blue, anchor=south]      at (pB) {$p_B=\left(\tfrac12,\tfrac{1}{2\sqrt3}\right)$};
        
        \draw[->, thick] (O) -- (U1)
            node[pos=1, below right] {$a_1=(1,0)$};
        \draw[->, thick] (O) -- (U2)
            node[pos=1, above left] {$a_2=\left(\tfrac12,\tfrac{\sqrt3}{2}\right)$};
        
        \node[black, anchor=west, scale=0.75] at ($ (pA) + (0.02,0.02) $) {$A_{0,0}$};
        \node[black, anchor=west, scale=0.75] at ($ (pB) + (0.02,0.02) $) {$B_{0,0}$};
        
        \coordinate (A10) at ($ (pA) + 1*(a1) $);
        \coordinate (B10) at ($ (pB) + 1*(a1) $);
        \node[black, anchor=west, scale=0.75] at ($ (A10) + (0.02,0.02) $) {$A_{1,0}$};
        \node[black, anchor=west, scale=0.75] at ($ (B10) + (0.02,0.02) $) {$B_{1,0}$};
        
        \coordinate (A01) at ($ (pA) + 1*(a2) $);
        \coordinate (B01) at ($ (pB) + 1*(a2) $);
        \node[black, anchor=west, scale=0.75] at ($ (A01) + (0.02,0.02) $) {$A_{0,1}$};
        \node[black, anchor=west, scale=0.75] at ($ (B01) + (0.02,0.02) $) {$B_{0,1}$};    
    \end{tikzpicture}
    \caption{An illustration of a honeycomb structure described in Example \ref{ex:honeycomb}.}
    \label{fig:honeycomb-aligned}
\end{figure}

\begin{example}\label{ex:honeycomb}
 The honeycomb graph in Figure \ref{fig:honeycomb-aligned} has a nice mathematical derivation. Let $a_1 = (1,0)$, $a_2 = \left(\frac{1}{2},\frac{\sqrt{3}}{2}\right),$
    and define the lattice $\Lambda = a_1\mathbb{Z}\oplus a_2\mathbb{Z}\subset \mathbb{Z}^2,$ where $a_i\mathbb{Z}\coloneqq \{a_i\times n : n\in \mathbb{Z} \},$
    for $i=1,2$. Let $p_A = (0,0)$, $p_B = \left(\frac{1}{2}, \frac{1}{2\sqrt{3}} \right),$ and define the vertex set $V = \{A_{m,n}, B_{m,n} : m,n\in\mathbb{Z}\},$ where $A_{m,n} \coloneqq p_A + ma_1 + na_2$ and $B_{m,n}\coloneqq p_B + ma_1 + na_2.$

    Then, the honeycomb lattice is obtained by the following adjacent vertices:
    \begin{align*}
        A_{m,n} \sim &\, B_{m,n},\\
        A_{m,n} \sim &\, B_{m-1,n},\\
        A_{m,n} \sim &\, B_{m,n-1},
    \end{align*}
    where $\sim$ means there is an edge between the two, for all $m,n\in\mathbb{Z}$. 
    We observe the fundamental domain is $Y = \{sa_1 + ta_2 : 0\leq s,t<1 \}$ and thus the unit cell is $V_Y = V \cap Y = \{p_A, p_B\}.$
\end{example}

Finally, we end this subsection with the definition of equivalence, not only to clear up any ambiguity but also to establish a framework for results.  
\begin{definition}\label{def: space equiv}
    Two space graphs $(G_1,\Lambda_1)$ and $(G_2, \Lambda_2)$ are said to be space-graph equivalent and denoted by $(G_1,\Lambda_1)\sim (G_2,\Lambda_2)$ if there exists 
    \begin{enumerate}
        \item a graph isomorphism $\phi: G_1\to G_2$,
        \item and a lattice isomorphism $A:\Lambda_1\to\Lambda_2$, where for the graph isomorphism we have $\phi(v + \lambda) = \phi(v) + A(\lambda)$ for all $v\in V(G_1)$ and all $\lambda\in\Lambda_1$.       
    \end{enumerate}
    This notion of equivalence should preserve the periodic structure. 
\end{definition}

\begin{remark}\label{remark:graph-isomorphism-observations}
    One may show that the lattice isomorphism $A$ is a linear map; the construction comes from working with the linearly independent base elements and building the characteristics from there. Furthermore, given the collection of vertices that lie on the unit cell, say $\{v_1, \ldots, v_d\}$, for an arbitrary $v \in V_1$, there exists $\lambda \in \Lambda_1$ such that $v = v_i +  \lambda$, for some $i=1,\dots, d$. Therefore, all vertices in $V_2$ can be written as $\phi(v_i) + A(\lambda)$ for some $i$ and $\lambda \in \Lambda_1$. 
\end{remark}

\section{Neighborhoods for Periodic Graphs}\label{section: periodic nbhds}
The structure of a crystal graph has a natural periodicity. There is a lot of information within the periodicity and the analysis in the paper leverages this information. However, the structure discussed in the previous section is insufficient. To add further context and tools, this section discusses the mathematical structure of periodic graphs.

A periodic graph has the potential to be quite large, and in order to extract computational information at scale, one has to consider local characteristics. Of course, since the graph is periodic, there is meaningful information that extends to the global graph, particularly, that of the fundamental domain. The following notation and definitions give us the rigorous foundation for the local analysis.  

\begin{definition}
    Let $V_Y$ be a unit cell with respect to fundamental domain $Y$. The diameter is defined as 
    \begin{equation}\label{eq: diameter}
        D_Y\coloneqq \max_{u,v\in V_Y}d_G(u,v).
    \end{equation}
\end{definition}
With this definition, we can define a $h$-hop ball.

\begin{definition}
Let $G$ be a space graph. For $v\in V$, define the $h$-hop ball as
\begin{equation}\label{eq: local ball}
    B_h^G(v) \coloneqq \{u \in V : d_G(u,v)\leq h\}.
\end{equation}
\end{definition}

Before our next result, we need a way to ensure that the unit cell only connects to nearby cells and not arbitrary distant ones (similar to how atoms interact with only nearby atoms). To do this, define the edge-span by
    \begin{equation}\label{eq: edge-span}
        \Delta(Y) \coloneqq\{\lambda\in\Lambda : \exists \,u,w\in V_Y \text{ with } (u,w+\lambda)\in E\}. 
    \end{equation}
    In our work, we only consider crystals where the edge-span cardinality, i.e., $| \Delta(Y)|$, is finite. 
    For example, the square lattice only allows edges to connect by translating left, right, up, and down. Then
    $\Delta(Y) = \{(1,0), (-1,0), (0,1), (0,-1)\}$. See \Cref{tikz:square_delta} for its visual. 

\begin{figure}[htbp]
    \centering

    \begin{subfigure}[t]{0.48\linewidth}
        \centering
        \begin{tikzpicture}[scale=0.85]
            \foreach \x in {-1,0,1}{
                \foreach \y in {-1,0,1}{
                    \fill[blue] (\x,\y) circle (2.5pt);
                }
            }

            \draw[thick,black] (-1,0)--(0,0)--(1,0);
            \draw[thick,black] (0,-1)--(0,0)--(0,1);

            \fill[red] (0,0) circle (3pt);

            \draw[->,thick,green!60!black] (0,0) -- (1,0);
            \draw[->,thick,green!60!black] (0,0) -- (-1,0);
            \draw[->,thick,green!60!black] (0,0) -- (0,1);
            \draw[->,thick,green!60!black] (0,0) -- (0,-1);
        \end{tikzpicture}

        \caption{$\Delta(Y)=
        \{(1,0),(-1,0),(0,1),(0,-1)\}$.}
        \label{tikz:square_delta}
    \end{subfigure}
    \hfill
    \begin{subfigure}[t]{0.48\linewidth}
        \centering
        \begin{tikzpicture}[scale=0.60]
            \foreach \x in {-3,-2,-1,0,1,2,3}{
                \fill[blue] (\x,0) circle (2.5pt);
            }

            \fill[red] (0,0) circle (3pt);

            \foreach \x in {-3,-2,-1,1,2,3}{
                \draw[thick,black] (0,0)--(\x,0);
            }

            \draw[->,thick] (3,0) -- (3.8,0);
            \draw[->,thick] (-3,0) -- (-3.8,0);
        \end{tikzpicture}

        \caption{$\Delta(Y)=
        \{(k,0):k\in\mathbb{Z}\}$.}
        \label{tikz:inf_delta}
    \end{subfigure}

    \caption{Examples of finite and infinite displacement sets
    $\Delta(Y)$.}
\end{figure}
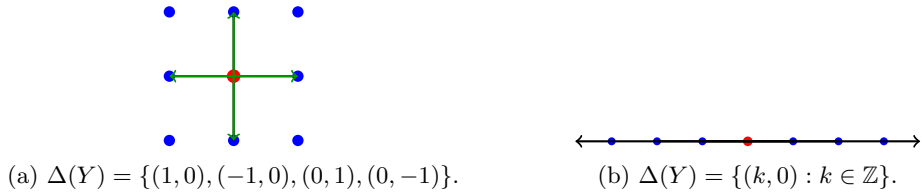

    To see an infinite edge-span example, consider some periodic graph with $\Lambda = \mathbb{Z}$ lattice. Take one vertex in its unit cell $v\in V_Y$, and then connect it to infinitely many translated copies $v\to v+k$, for every $k\in\mathbb{Z}$. Then $\Delta(Y) = \mathbb{Z}$. See \Cref{tikz:inf_delta} for a visual aid.

\begin{theorem}\label{theorem: unit cell translation}
    Let $G$ be a connected space graph with respect to lattice $\Lambda$. Fix a fundamental domain $Y$. Then, there exists $\tilde{h}\in \mathbb{N}$ such that for every $h\geq \tilde{h}$, there exists $\lambda\in \Lambda$ for which the ball $B_h^G(v)$ satisfies the following:
    \begin{equation*}
        \lambda + V_Y\subset B_h^G(v).
    \end{equation*}
\end{theorem}

\begin{proof}
    For the fundamental domain $Y$, the unit cell connects to finitely many neighboring cells. That is in this work, we assume the set 
    \begin{equation*}
        \Delta(Y) \coloneqq\{\lambda\in\Lambda : \exists \,u,w\in V_Y \text{ with } (u,w+\lambda)\in E\}
    \end{equation*}
    is finite.  
    Define
    \begin{equation}\label{eq: S}
        S\coloneqq V_Y \cup \bigcup\limits_{\lambda\in\Delta(Y)}(V_Y+\lambda)\subset V.
    \end{equation}
    This set is finite because $V_Y$ is finite and so is $\Delta(Y)$. Fix a reference vertex $v_0\in V_Y$, and  set 
    \begin{equation*}
        h_0\coloneqq \max\limits_{x\in S}d_G(v_0,x). 
    \end{equation*}
    This maximum is indeed finite because $S$ is finite and $G$ is connected. Then by definition, we have that $S\subseteq B_{h_0}^G(v_0)$. 
    Let $v\in V$ be an arbitrary vertex. Since $Y$ is a fundamental domain, it follows that there exists $\mu\in\Lambda$ and $u\in V_Y$ such that $v=u+\mu$. By \Cref{prop: graph auto}, translations by $\Lambda$ are graph automorphisms and hence they preserve graph distances. i.e., for any $u,v\in V$ and for any $\mu\in\Lambda$, we have that $d_G(u+\mu,v+\mu) = d_G(u,v). $ Hence $S+\mu\subseteq B_{h_0}^G(v_0+\mu).$ In particular, since $h\geq h_0$ implies that $B_{h_0}^G(\cdot)\subseteq B_h^G(\cdot)$, resulting in $S+\mu\subseteq B_{h}^G(v_0+\mu),$ for all $h\geq h_0$. 
    
    For any vertex $x\in V$ and by the triangle inequality, we see that 
   
    \begin{align*}
        d_G(x,v)\leq & \, d_G(x, v_0+\mu) + d_G(v_0+\mu,v) \\
        = & \, d_G(x, v_0 +\mu) + d_G(v_0,u) \\
        \leq & \, h_0 + D_Y. 
    \end{align*}
    Thus, every $w\in S+\mu$ satisfies $d_G(w,v)\leq h_0+D_Y,$ so that $S+\mu\subseteq B_{h+D_Y}^G(v).$

    Setting $\tilde{h}\coloneqq h_0+D_Y$, we see that for all $v\in V$ and every $h\geq \tilde{h}$, there exists $\mu\in \Lambda$ such that $S+\mu \subseteq B_h^G(v)$, as desired.
\end{proof}

\Cref{theorem: unit cell translation} explains why local neighborhoods eventually contain complete motif information and thus provides the geometric foundation for the local invariants introduced in \Cref{section: local invariant}. From this, we highlight several pairs of graph examples in \Cref{tab:comparison} below. 

\begin{remark}
    In the previous proof, we used the set $S$ defined in \Cref{eq: S} rather than just $V_Y$. Using $V_Y$ is not enough to capture the motif of the graph because the motif contains the vertices $V_Y$ as well as the edges inside the cell and edges from the unit cell to neighboring cells. The assumption that $\Delta(Y)$ is finite is to ensure that we only consider neighboring interactions. 
\end{remark}

\section{Local invariant}\label{section: local invariant}
Using the tools we established for graphs, we move onto the local structure of graphs from the perspective of hops and local adjacency matrices. Interestingly, adjacency matrices have a natural Hermitian structure with zero trace. These are exactly the matrices that generate the Lie algebra of special unitary operators, modulo $\sqrt{-1}$. For asymmetric adjacency matrices, one may apply a Hermitian adjacency matrix transform. From this observation, we will explore what information one may extract when working from this perspective.

Graphs are extremely complex structures that can be exponential in size, which makes it computationally difficult to extract latent information, in general. Consequently, one may focus on analyzing the local structure of nodes. Of course, with periodic graphs, one must analyze the local structure. 

Particularly, the focus is on information retention of local adjacency matrices around nodes. In this section, we formalize this concept. For the remainder of this section and paper, for every local neighborhood $B_h^G(v)$, we assume a fixed canonical vertex ordering that depends only on the chosen vertex $v$ and the graph induced by its neighborhood.   

\begin{definition}
    From~\cite{vlasic2024quop}, and using \Cref{eq: local ball}, let $G$ be a graph and let $v\in V$. Let $G[B_h^G(v)]$ denote the induced subgraph on this ball. Fix an ordering of vertices mentioned above and let 
    \begin{equation}
        A^G_h(v)\in M_{m}(\mathbb{C}),\qquad m\coloneqq |B_h^G(v)|
    \end{equation}
    be the local adjacency matrix of $G[B_h^G(v)]$ with respect to this ordering, where $M_{m}(\mathbb{C})$ denotes the set of square matrices of dimension $m$ with complex values. 
\end{definition}
For each graph $G$, fixed vertex $v$, and hop $h$, we form the skew Hermitian matrix
\begin{equation}\label{eq: hamiltonian}
    H^G_h(v) \coloneqq -\mathrm{i}\,A^G_h(v).
\end{equation}
After embedding $H_h^G(v)$ into a $2^n\times 2^n$ matrix (via padding the matrix with zeros if needed), where $n = \left\lceil \log_2 \left(|B_h^G(v)| \right)\right\rceil$, 
let 
\begin{equation}\label{eq:pauli-ortho-basis}
    \mathcal{P}_n\coloneqq\left\{ \bigotimes_{j=0}^{n-1} \sigma_j
    :
    \sigma_j\in\{I,X,Y,Z\} : j=0,1,\dots, n-1
    \right\}
\end{equation}
denote the n-qubit Pauli basis. For $\bm\sigma = (\sigma_0,\dots \sigma_{n-1})\in \{I,X,Y,Z\}^n$,we write 
\begin{equation}
    P_{\bm\sigma}=\bigotimes_{j=0}^{n-1} \sigma_j\in\mathcal{P}_n. 
\end{equation}
In this basis, any n-qubit operator admits an expansion:
\begin{equation}\label{eq: pauli basis operator}
    H^G_h(v)=\sum\limits_{\bm\sigma\in\{I,X,Y,Z\}^n}c^G_h(v,\bm\sigma)\times P_{\bm\sigma},
\end{equation}
where $c^G_h(v,\bm\sigma)=\frac{1}{2^n}\Tr(P^*_{\bm\sigma}H^G_h(v))$. This follows from the fact that Pauli operators form an orthonormal basis for $\mathrm{End}(\mathcal{H}_n)$, the algebra of all linear operators on the Hilbert space $\mathcal{H}_n$, with respect to the Hilbert-Schmidt inner product $\langle A,B\rangle \coloneqq \frac{1}{2^n}\Tr(A^*B)$. 

For a fixed radius $h\in\mathbb{N}$, we define its \textbf{root set} at hop size $h$ to be 
\begin{equation}\label{eq: roots}
    R^G_h(v)\coloneqq\bigl\{ \bm{\sigma}\in\{I,X,Y,Z\}^n :
    c^G_h(v;\bm{\sigma})\neq 0 \bigr\}.
\end{equation}
The following proposition describes invariance with respect to permutation, an essential characteristic for equivalence when considering algebraic objects. 

\begin{proposition}\label{prop: root translation invariant}
    For all $\lambda\in\Lambda, h\in\mathbb{N}$, and $v\in V$, we have
    \begin{equation}
        R_h^G(v) = R_h^G(v+ \lambda). 
    \end{equation} 
\end{proposition}

\begin{proof}
    Fix $\lambda\in\Lambda$ and $v\in V$. By \Cref{prop: graph auto}, the map $\tau_\lambda$ is a graph automorphism. In particular, it preserves graph distances and thus
    $$\tau_\lambda(B_h^G(v)) = B_h^G(\tau_\lambda(v)).$$ 
    This implies that $\tau_\lambda$ induces a graph isomorphism and hence 
    $$G[B_h^G(v)] \cong G[B_h^G(v+\lambda)],$$ 
    for all $\lambda\in\Lambda$.
    
    Let $A_h^G(v)$ and $A_h^G(v+\lambda)$ denote the adjacency matrices of $G[B_h^G(v)]$ and $G[B_h^G(v+\lambda)]$ respectively. We assume these matrices are written using a fixed vertex ordering that was also used in the construction of $R_h^G(\cdot)$. Since the graphs $G[B_h^G(v)] $ and $G[B_h^G(v+\lambda)]$ are isomorphic, and by the canonical ordering fixed above, it follows that their skew-Hermitian adjacency matrices are identical. Since Pauli expansions are unique, the corresponding Pauli-support sets are equal. Namely
    \begin{equation*}
        R_h^G(v)=R_h^G(v+\lambda)
    \end{equation*}
    as desired. 
\end{proof}

For the following theorem we write 
$$R_h^G\coloneqq \bigcup\limits_{v\in V}R_h^G(v)$$  
to capture the global collection of all roots. Here, we assume that repeated entries are considered as a single element. What we would like to be true is that $R_h^G$ holds latent information on isomorphically connecting $G$ to other graphs. i.e., two graphs are isomorphic if and only if their global roots set are equal up to some permutation.  However, in general, the statement does not hold and what we are able to establish is the sufficient direction.    

\begin{theorem}
    Let $(G_1,\Lambda_1)\sim(G_2,\Lambda_2)$ be equivalent in the sense of \Cref{def: space equiv}. Then, for every $h>0$, we have
    \begin{equation}
        R_h^{G_1} = R_h^{G_2}.
    \end{equation}
\end{theorem}

\begin{proof}
    Let $\phi:G_1\to G_2$ and $A :\Lambda_1\to\Lambda_2$ be the graph isomorphism and the lattice isomorphism from the space graph equivalence with $\phi(\lambda + v)=A (\lambda) + \phi(v)$. Since $\phi$ is a graph isomorphism, it preserves graph distances. Thus, for every $v\in V$, we have
    \begin{equation*}
        \phi(B_h^{G_1}(v)) = B_h^{G_2}(\phi(v)).
    \end{equation*}
    Moreover, the isomorphism also implies that
    \begin{equation}\label{eq: local graph cong}
        G_1[B_h^{G_1}(v)] \cong G_2[B_h^{G_2}(\phi(v))].
    \end{equation}
    From this and by the canonical ordering convention fixed above, the corresponding ordered local adjacency matrices are identical. Since Pauli expansions are unique, it follows that 
    \begin{equation*}
        R_h^{G_1}(v)=R_h^{G_2}(\phi(v)).
    \end{equation*}
    Taking a union over all $v\in V_1$ and using the fact that $\phi$ is a bijection from $V_1\to V_2$, we have that 
    \begin{align*}
        R_h^{G_1} 
         = 
        \bigcup\limits_{v\in V_1} R_h^{G_1}(v) 
         = \bigcup\limits_{v\in V_1} R_h^{G_2}(\phi(v)) 
         = \bigcup\limits_{w\in V_2}R_h^{G_2}(w)
         = R_h^{G_2}
    \end{align*}
    as desired. 
\end{proof}

\section{Root Separation and Local Converse Results}\label{section:converse}

In this section, we take the local structure around a node that was explored in the previous, described by the set of orthonormal basis elements in Equation \eqref{eq:pauli-ortho-basis} that are obtained from a ball of hop size $h$. 

\begin{definition}
    Let $(G,\Lambda)$ be a connected periodic graph with finite unit cell and finite edge-span.
    We say that $G$ is root-separated at radius $h$ if for all $v,w\in V(G)$,
    \begin{equation}
        B_h^G(v)\not\cong B_h^G(w)
        \implies
        R_h^G(v)\neq R_h^G(w).
    \end{equation}
\end{definition}

\begin{definition}
    Let $(G,\Lambda)$ be a connected periodic graph with finite unit cell and finite edge-span. Define the eventual root separation radius by
    \begin{equation}
        \rho_{\mathrm{ev}}(G) \coloneqq \inf\left\{ h_\ast\in\mathbb{N} : G \text{ is root-separated at every radius } h\ge h_\ast \right\},
    \end{equation}
    with the convention that $\rho_{\mathrm{ev}}(G)=\infty$ if no such $h_\ast$ exists.
\end{definition}
This quantity measures the scale at which the Pauli-support invariant becomes complete for distinguishing local rooted types.
The significance of the next  proposition is that root separation, which is a priori a global condition over the infinite vertex set $V(G)$, can be verified by checking only finitely many pairs of vertices in a single unit cell. Thus, for periodic graphs, root separation reduces to a finite combinatorial verification problem.

\begin{proposition}
    Let $(G,\Lambda)$ be a connected periodic graph with finite unit cell and finite edge-span. Fix $h\in\mathbb{N}$ and let $V_Y$ be a unit cell.

    Then $G$ is root-separated at radius $h$ if and only if
    \begin{equation}\label{eq: root sep unit cell}
        B_h^G(v)\not\cong B_h^G(w) \implies R_h^G(v)\neq R_h^G(w)
    \end{equation}
    for all $v,w\in V_Y$.
\end{proposition}

\begin{proof}
    Assume that $G$ is root-separated at radius $h$. By definition, \Cref{eq: root sep unit cell} is true for all for all $x,y\in V(G)$ and hence for $V_Y\subset V(G)$. 

    Conversely, let $x,y\in V(G)$. Since $(G,\Lambda)$ is periodic, there exist $v,w\in V_Y$ and $\lambda,\mu\in\Lambda$ such that $x = v+\lambda$ and $y = w+\mu.$ By translation invariance of the graph, we have $B_h^G(x)\cong B_h^G(v)$ and $B_h^G(y)\cong B_h^G(w).$ Particularly,
    \begin{equation*}
        B_h^G(x)\not\cong B_h^G(y) \implies     B_h^G(v)\not\cong B_h^G(w).
    \end{equation*}
    By the assumption on $V_Y$, it follows that $R_h^G(v)\neq R_h^G(w).$ By \Cref{prop: root translation invariant}, this implies that we have $R_h^G(x)=R_h^G(v)$ and $R_h^G(y)=R_h^G(w)$.

    Therefore, $R_h^G(x)\neq R_h^G(y),$ which proves that $G$ is root-separated at radius $h$ as desired. 
\end{proof}

\begin{theorem}
    Let $(G, \Lambda)$ be a connected periodic graph with finite unit cell and finite edge-span. Fix $h\in \mathbb{N}$ and assume $G$ is root-separated at radius $h$. Then for all $v,w\in V_G$, we have
    \begin{equation}
        R_h^G(v)=R_h^G(w) \implies  B_h^G(v)\cong B_h^G(w).
    \end{equation}
\end{theorem}

\begin{proof}
    Let $v,w\in V_G$ and assume $R_h^G(v)=R_h^G(w)$. Choose $v_0,w_0\in V_Y$ and $\lambda, \mu\in\Lambda$ such that 
    \begin{equation}
        v=v_0+\lambda, \qquad w=w_0+\mu.
    \end{equation}
    By translation invariance of the graph, we have that 
    \begin{equation}\label{eq: ball cong}
        B_h^G(v)\cong B_h^G(v_0), \qquad B_h^G(w)\cong B_h^G(w_0).
    \end{equation}
    Since the ordering of the vertices is canonical, \Cref{eq: ball cong} implies that the padded adjacency matrices are identical. i.e., 
    \begin{equation}
        A_h(v) = A_h(v_0) \qquad A_h(w)=A_h(w_0).
    \end{equation}
    Hence the corresponding root sets satisfy 
    \begin{equation}
        R_h^G(v_0)=R_h^G(w_0).
    \end{equation}
    By assumption, the graph $G$ is root-separated at radius $h$, so it follows that 
    \begin{equation}
        B_h^G(v_0)\cong B_h^G(w_0).
    \end{equation}
    By \Cref{eq: ball cong}, this implies that 
    \begin{equation}
        B_h^G(v)\cong B_h^G(w)
    \end{equation}
    as desired. 
\end{proof}

\section{Preservation of Lie-algebraic structures}\label{section: Lie}
We move onto taking the set of root elements and consider the Lie-algebraic closure of this set, asking the question about what information is retained with these dynamic Lie algebras. Simulations on quantum computers are a promising technique for understanding quantum dynamics at large scale. Algorithms that utilize Lie algebraic tools, such as Cartan decompositions, offer speed ups and deeper understandings of systems under investigation~\cite{kokcu2022fixed}. Before proving our results on how our root set constructions are preserved at the Lie algebraic level, we provide an introductory overview on the necessary Lie algebraic tools needed in this manuscript. 

\subsection{A Brief Overview of Lie Algebras}
This subsection contains the preliminaries for \Cref{subsubsection: Lie}. 

The Lie algebra of the unitary group is denoted as 
$$\mathfrak{u}(d)\coloneqq \{X\in M_{d}(\mathbb{C}) : X^\star = -X\}.$$ 
i.e., the real vector space of all $d\times d$ skew-Hermitian matrices equipped with the commutator $[X,Y]\coloneqq XY-YX$.  Let $S\subset\mathfrak{u}(d)$. We denote $\mathrm{Lie}(S)$
for the smallest Lie subalgebra of $\mathfrak{u}(d)$ containing $S$. Equivalently, $\mathrm{Lie}(S)$ is obtained from $S$ by repeatedly taking linear combinations of commutators.

In quantum mechanics, the Lie group captures the symmetric group of physical systems, and the Lie algebra is the tangent space of this group, which is centered at the identity element \cite{wigner2012group,varadarajan2013lie,bincer2013lie}. 
This structure gives rise to the following (Theorem 3.20 in Hall \cite{hall2013lie}):

\begin{enumerate}
    \item $AXA^{-1} \in \mathfrak{g}$ for $X \in \mathfrak{g}$ and $A \in \mathbb{G}$, where $\mathbb{G}$ is a Lie matrix group, and $\mathfrak{g}$ is the associated Lie algebra;
    \item $tX \in \mathfrak{g}$ for $t\in \mathbb{R}$;
    \item $X + Y \in \mathfrak{g}$;
    \item and commutator $[X,Y] := XY - YX \in \mathfrak{g}$.
\end{enumerate}

From~\cite{hall2013lie}, recall the next 3 definitions to establish a foundation for the KAK decomposition, an essential characteristic to simplify the analysis and computation.   

\begin{definition}
    If $\mathfrak{g}$ is a semisimple Lie algebra, then a \underline{Cartan subalgebra} of $\mathfrak{g}$ is a subspace $\mathfrak{h}\subset \mathfrak{g}$ with the following properties:
    \begin{enumerate}
        \item $[H_1,H_2]=0$, for all $H_1,H_2\in\mathfrak{h}$
        \item whenever  $X\in\mathfrak{g}$ satisfies $[H,X]=0$, for all $H\in\mathfrak{h}$, then $X\in\mathfrak{h}$
        \item For all $H\in\mathfrak{h}$, the linear map $\operatorname{ad}_H:\mathfrak{g}\to\mathfrak{g}$, where $\operatorname{ad}_H(X)\coloneqq [H,X]$ is diagonalizable.
    \end{enumerate}
\end{definition}

Condition 1 says that $\mathfrak{h}$ is a commutative subalgebra of $\mathfrak{g}$, condition 2 says that $\mathfrak{h}$ is a maximal commutative subalgebra, and condition 3 says that each $\Ad_H$ and $H\in\mathfrak{h}$ are simultaneously diagonalizable. See Proposition 7.11 in~\cite{hall2013lie} for the existence of Cartan subalgebras.

\begin{definition}
     A \underline{Cartan decomposition} of a semisimple Lie algebra $\mathfrak{g}$ is a vector space decomposition
    \begin{equation}\label{eq:Cartan-decomp}
        \mathfrak{g}\;=\;\mathfrak{k}\  \oplus\ \mathfrak{p}
    \end{equation}
   and bracket relations
    \begin{equation*}
        [\mathfrak{k},\mathfrak{k}]\subseteq \mathfrak{k},\qquad
        [\mathfrak{k},\mathfrak{p}]\subseteq \mathfrak{p},\qquad
        [\mathfrak{p},\mathfrak{p}]\subseteq \mathfrak{k}.
    \end{equation*}
    Here, $\mathfrak{k}$ is the Lie subalgebra fixed by an involution $\theta:\mathfrak{g}\to\mathfrak{g}$ so that 
    \begin{equation*}
        \mathfrak{k}= \{X\in \mathfrak{g}: \theta(X)=X \}.
    \end{equation*}
    From this, then $\mathfrak{p}$ is simply the $-1$ eigenspace of $\theta$. That is 
    \begin{equation}
        \mathfrak{p}= \{X\in \mathfrak{g}: \theta(X)= -X \}.
    \end{equation}
\end{definition}

\begin{definition}
    Let $\mathcal{G}$ be a connected Lie group with Lie algebra $\mathfrak{g}$, and let 
    \begin{equation}
         \mathfrak{g}\;=\;\mathfrak{k}\  \oplus\ \mathfrak{p}
    \end{equation}
    be a Cartan decomposition induced by a Cartan involution $\theta$. Let $\mathcal{K}\subset\mathcal{G}$ be the connected Lie group with Lie algebra $\mathfrak{k}$, and let $\mathfrak{a}\subset\mathfrak{p}$ be a Cartan subspace. Define $\mathcal{A}\coloneqq\exp(\mathfrak{a})$. The \underline{$KAK$ decomposition} states that every element $g\in\mathcal{G}$ can be written as
    \begin{equation}
        g = k_1\,a\,k_2,
    \end{equation}
    where $k_1,k_2\in\mathcal{K}$ and $a\in\mathcal{A}$. 
    Equivalently, 
    \begin{equation}
        \mathcal{G}=\mathcal{KAK}.
    \end{equation} 
\end{definition}
The KAK decomposition is the global Lie group analogue of the Cartan decomposition at the Lie algebra.

\subsection{Lie Algebraic Main Results}\label{subsubsection: Lie}
The results in this paper leverage characteristics of the Lie algebra, and subsequent Lie group, of special unitary matrices. 

\begin{proposition}\label{prop: lie closures for roots}
    Let $U\in U(d)$ be unitary and let $S\subset \mathfrak{u}(d)$.
    Then
    \begin{equation}\label{eq:Ad-Lie}
        \mathrm{Lie}(USU^*) \;=\; U\,\mathrm{Lie}(S)\,U^*.
    \end{equation}
    Equivalently, writing $\Ad_U(X)\coloneqq UXU^*$, one has
    $\mathrm{Lie}(\Ad_U(S))=\Ad_U(\mathrm{Lie}(S))$.
\end{proposition}
\begin{proof}
    First we show $\Ad_U$ is a Lie algebra automorphism of $\mathfrak{u}(d)$. Linearity is clear. To show $\Ad_U$ preserves the Lie bracket, for $X,Y\in\mathfrak{u}(d)$,
    \begin{align*}
    \Ad_U([X,Y])
    &= U(XY-YX)U^* \\
    &= UXYU^* - UYXU^* \\
    &= (UXU^*)(UYU^*) - (UYU^*)(UXU^*) \qquad (\text{since }U^*U=I)\\
    &= [\Ad_U(X),\Ad_U(Y)].
    \end{align*}

    Next, we show $\Ad_U$ is bijective with inverse $\Ad_{U^*}$.
    For any $X\in\mathfrak{u}(d)$,
    \begin{equation}
    \Ad_U(\Ad_{U^*}(X))=U(U^* X U)U^*=(UU^*)X(UU^*)=X,
    \end{equation}
    and similarly $\Ad_{U^*}(\Ad_U(X))=X$. Hence $\Ad_U$ is bijective and $\Ad_U^{-1}=\Ad_{U^*}$.

    Let $\mathfrak{g}\coloneqq \mathrm{Lie}(S)$.
    By definition, $\mathfrak{g}$ is the smallest Lie subalgebra of $\mathfrak{u}(d)$ that contains $S$.
    Since $\Ad_U$ is a Lie algebra automorphism, the image $\Ad_U(\mathfrak{g})=U\mathfrak{g}U^*$ is again a Lie subalgebra of $\mathfrak{u}(d)$.
    Moreover, because $S\subset \mathfrak{g}$, we have
    \begin{equation}
    USU^*=\Ad_U(S)\subset \Ad_U(\mathfrak{g})=U\mathfrak{g}U^*.
    \end{equation}
    Thus $U\mathfrak{g}U^*$ is a Lie subalgebra containing $USU^*$, so by minimality,
    \begin{equation}
    \mathrm{Lie}(USU^*) \subseteq U\,\mathrm{Lie}(S)\,U^*.
    \end{equation}
    
    For the reverse inclusion, apply the same argument to $U^*$ in place of $U$.
    Starting from the inclusion just proved with $U^*$,
    \begin{equation}
        \mathrm{Lie}(U^*(USU^*)U) \subseteq U^*\,\mathrm{Lie}(USU^*)\,U.
    \end{equation}
    But $U^*(USU^*)U=S$, hence the left-hand side is $\mathrm{Lie}(S)$, so
    \begin{equation}
        \mathrm{Lie}(S)\subseteq U^*\,\mathrm{Lie}(USU^*)\,U.
    \end{equation}
    Conjugating both sides by $U$ yields
    \begin{equation}
        U\,\mathrm{Lie}(S)\,U^* \subseteq \mathrm{Lie}(USU^*).
    \end{equation}
    Together with the first inclusion, this proves \Cref{eq:Ad-Lie}.
\end{proof}

For graphs, replace $S$ by a root set with respect to hop size $h$, then \Cref{prop: lie closures for roots} shows 
    \begin{equation}
        \overline{P}\,\mathrm{Lie}(R_h^{G})\,\overline{P}^*=\mathrm{Lie}(R_h^{\overline{G}}),
    \end{equation}
    where $\overline{P}$ is a permutation matrix, and $\overline{G}$ denotes $\overline{P}G\overline{P}^\star$.

The previous proposition shows that Lie closure is equivariant under unitary conjugation. Applying this to root sets, we can obtain that Cartan decompositions and KAK factorizations associated to the Lie algebra generated by root sets of graphs are also preserved. This is formalized in the following theorem.

\begin{theorem}\label{theorem: cartan equivariance}
    Let $G$ be a graph and let $H_h^G(v)$ be the skew Hermitian matrix generated by the local adjacency matrix for a fixed vertex $v$. Let $R_h^G$ be the $h$-hop root set constructed from the local adjacency matrix. Let $\overline{P}$ be a permutation matrix and denote the permuted graph by $\overline{G}$. 
    Assume the root set construction is permutation equivariant, i.e.,
    \begin{equation}
        R_h^{\overline{G}} = \overline{P} R_h^G \overline{P}^*
    \qquad \text{for all } h .
    \end{equation}
    Let
    \begin{equation}
        \mathfrak{g} := \mathrm{Lie}(R_t^G) \qquad \overline{\mathfrak{g}} := \mathrm{Lie}(R_t^{\overline{G}})
    \end{equation}
    be the Lie algebras generated by the root sets $R_h^G$ and $R_h^{\overline{G}}$ respectively. Then the following holds:
        \begin{enumerate}
    \item
        Suppose $\mathfrak{g} = \mathfrak{k} \oplus \mathfrak{p}$ is a Cartan decomposition with Cartan involution $\theta$.  
        Then
        \begin{equation}
            \overline{\mathfrak{g}} = \overline{P}\mathfrak{k}\overline{P}^* \;\oplus\; \overline{P}\mathfrak{p}\overline{P}^*
        \end{equation}
        is a Cartan decomposition of $\overline{\mathfrak{g}}$ with Cartan involution
        \begin{equation}
        \overline{\theta} = \operatorname{Ad}_{\overline{P}} \circ \theta \circ \operatorname{Ad}_{\overline{P}}^{-1}.
        \end{equation}
    
    \item
        Let $K=\exp(\mathfrak{k})$, and let $\mathfrak{a}\subset\mathfrak{p}$ be a maximal abelian subalgebra with $A=\exp(\mathfrak{a})$, yielding a KAK decomposition
        \begin{equation}
            \exp(\mathfrak{g}) = K A K .
        \end{equation}
        Then
        \begin{equation}
            \exp(\overline{\mathfrak{g}}) = (\overline{P}K\overline{P}^*) (\overline{P}A\overline{P}^*) (\overline{P}K\overline{P}^*)
        \end{equation}
        is a KAK decomposition of $\exp(\overline{\mathfrak{g}})$.
    \end{enumerate}
\end{theorem}

\begin{proof}
    By \Cref{prop: lie closures for roots}, conjugation by $\overline{P}$ induces a Lie algebra isomorphism. Namely 
    \begin{equation}    
        \operatorname{Ad}_{\overline{P}} : \mathfrak{g} \to \overline{\mathfrak{g}}, \qquad X \mapsto \overline{P} X \overline{P}^* .
    \end{equation}

    Let $\theta$ be the Cartan involution associated with the decomposition
    $\mathfrak{g}=\mathfrak{k}\oplus\mathfrak{p}$. 
    Define 
    \begin{equation}
        \overline{\theta} = \operatorname{Ad}_{\overline{P}} \circ \theta \circ \operatorname{Ad}_{\overline{P}}^{-1}.
    \end{equation}
    Since $\operatorname{Ad}_{\overline{P}}$ is a Lie algebra automorphism and $\theta$ is a Cartan involution for $\mathfrak{g}$, it follows that $\overline{\theta}$ is also a Cartan involution for $\overline{\mathfrak{g}}$. 
    Moreover, the $\pm1$ eigenspaces of $\overline{\theta}$ are
    \begin{equation}
        \overline{\mathfrak{k}} := \overline{P}\mathfrak{k}\overline{P}^*, \qquad \overline{\mathfrak{p}} := \overline{P}\mathfrak{p}\overline{P}^*.
    \end{equation}
    Because $\operatorname{Ad}_{\overline{P}}$ preserves Lie brackets, the relations 
    \begin{equation}       
        [\mathfrak{k},\mathfrak{k}]\subset\mathfrak{k}, \quad [\mathfrak{k},\mathfrak{p}]\subset\mathfrak{p}, \quad [\mathfrak{p},\mathfrak{p}]\subset\mathfrak{k}
    \end{equation}
    are carried to the corresponding relations for $\overline{\mathfrak{k}}$ and $\overline{\mathfrak{p}}$.
    Hence 
    \begin{equation}
        \overline{\mathfrak{g}} = \overline{\mathfrak{k}} \oplus \overline{\mathfrak{p}}
    \end{equation}
    is the Cartan decomposition of  $\overline{\mathfrak{g}}$.

    Similarly for the KAK decomposition, let $\mathfrak{a}\subset\mathfrak{p}$ be a maximal abelian subalgebra and define $\overline{\mathfrak{a}} := \overline{P}\mathfrak{a}\overline{P}^*$. Since $\operatorname{Ad}_{\overline{P}}$ is a Lie algebra isomorphism, $\overline{\mathfrak{a}}$ is abelian and maximal in $\overline{\mathfrak{p}}$.
    Setting
    \begin{equation}
        \overline{K} := \exp(\overline{\mathfrak{k}}), \qquad \overline{A} := \exp(\overline{\mathfrak{a}}),
    \end{equation}
    we have 
    \begin{equation}
        \overline{K} = \overline{P}K\overline{P}^*, \qquad \overline{A} = \overline{P}A\overline{P}^*.
    \end{equation}
    Since conjugation by $\overline{P}$ is a group automorphism, it carries the KAK decomposition, 
    \begin{equation}
        \exp(\mathfrak{g}) = K A K
    \end{equation}
    to 
    \begin{equation}
        \exp(\overline{\mathfrak{g}}) = \overline{K}\,\overline{A}\,\overline{K}.
    \end{equation}
    Hence is a valid KAK decomposition of $\exp(\overline{\mathfrak{g}})$.
\end{proof}

\subsection{Empirical Experiments}\label{subsubsubsection: experiments}
To display the techniques, we calculate the root sets of the graphs given in Figure \ref{fig:iso_pair} displaying isomorphic graphs, Figure \ref{fig:square_triang} and Figure \ref{fig:cl_ml} showing  non-isomorphic graphs, and Figure \ref{fig:aperiodic} for simple examples of aperiodic graphs. The results are given in Table \ref{tab:comparison}.

The experiments were ran with Python where the graphs were generated with the Networkx package \cite{hagberg2020networkx} as well as to calculate the local adjacency matrices around a given node. To calculate the decomposition, supporting functions were written with the Numpy package \cite{harris2020array}. As the decompositions are with respect to a quantum circuit, local adjacency matrices that are not of the form $2^{n_0} \times 2^{n_0}$ were `padded' with zeros and the integer $n_0$ is taken as the largest pad value to ensure all adjacency matrices all have the same dimension. To determine the root decomposition, with respect to the Pauli operators, the root set is recursively calculated and the Frobenius norm is taken to determine if a root is in the latent sum.   

    \begin{figure}
        \centering
        \begin{subfigure}[t]{0.48\textwidth}
            \centering
            \includegraphics[height=4cm]{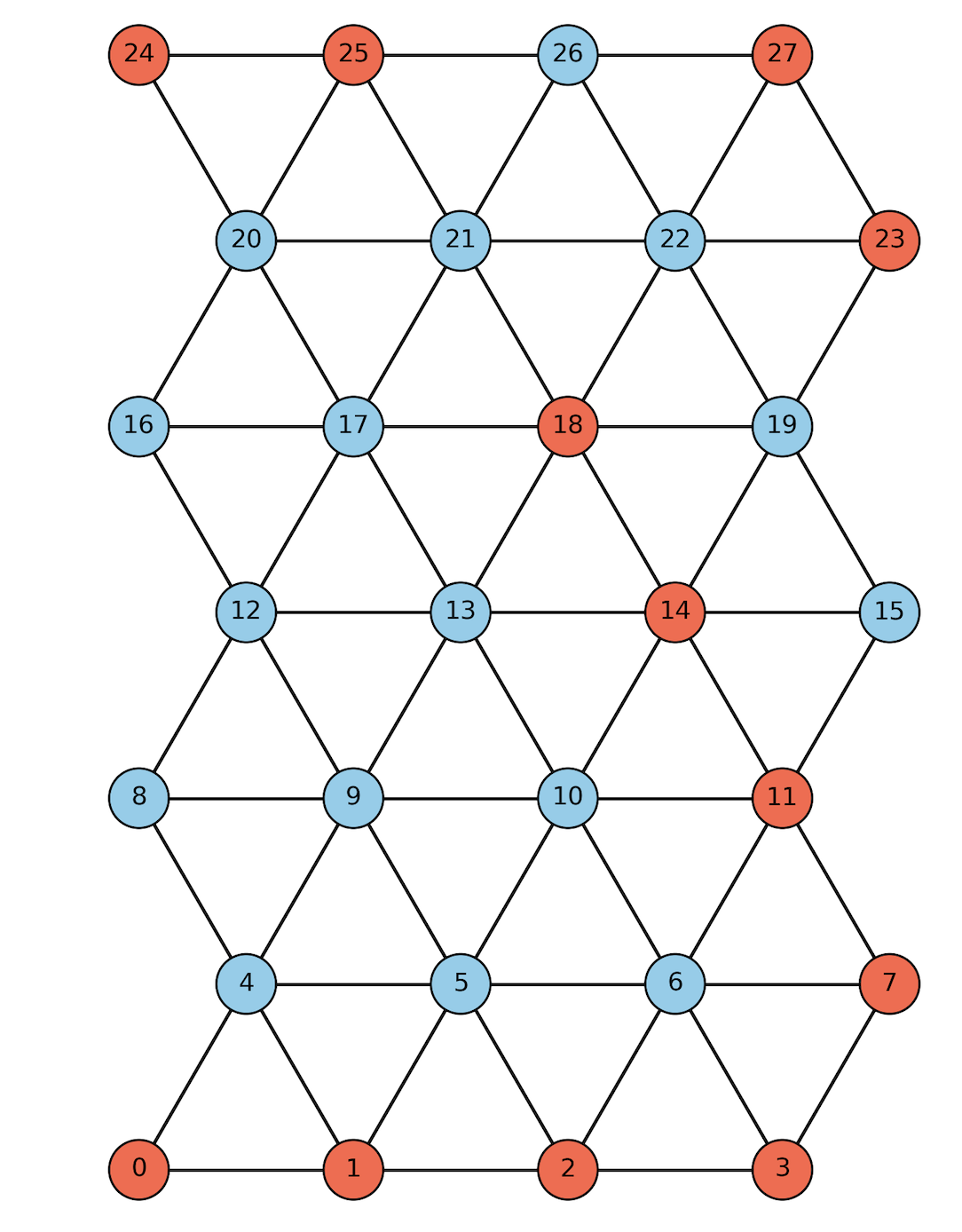}
            \caption{Triangular lattice example}
        \end{subfigure}
        \hfill
        \begin{subfigure}[t]{0.48\textwidth}
            \centering
            \includegraphics[height=4cm]{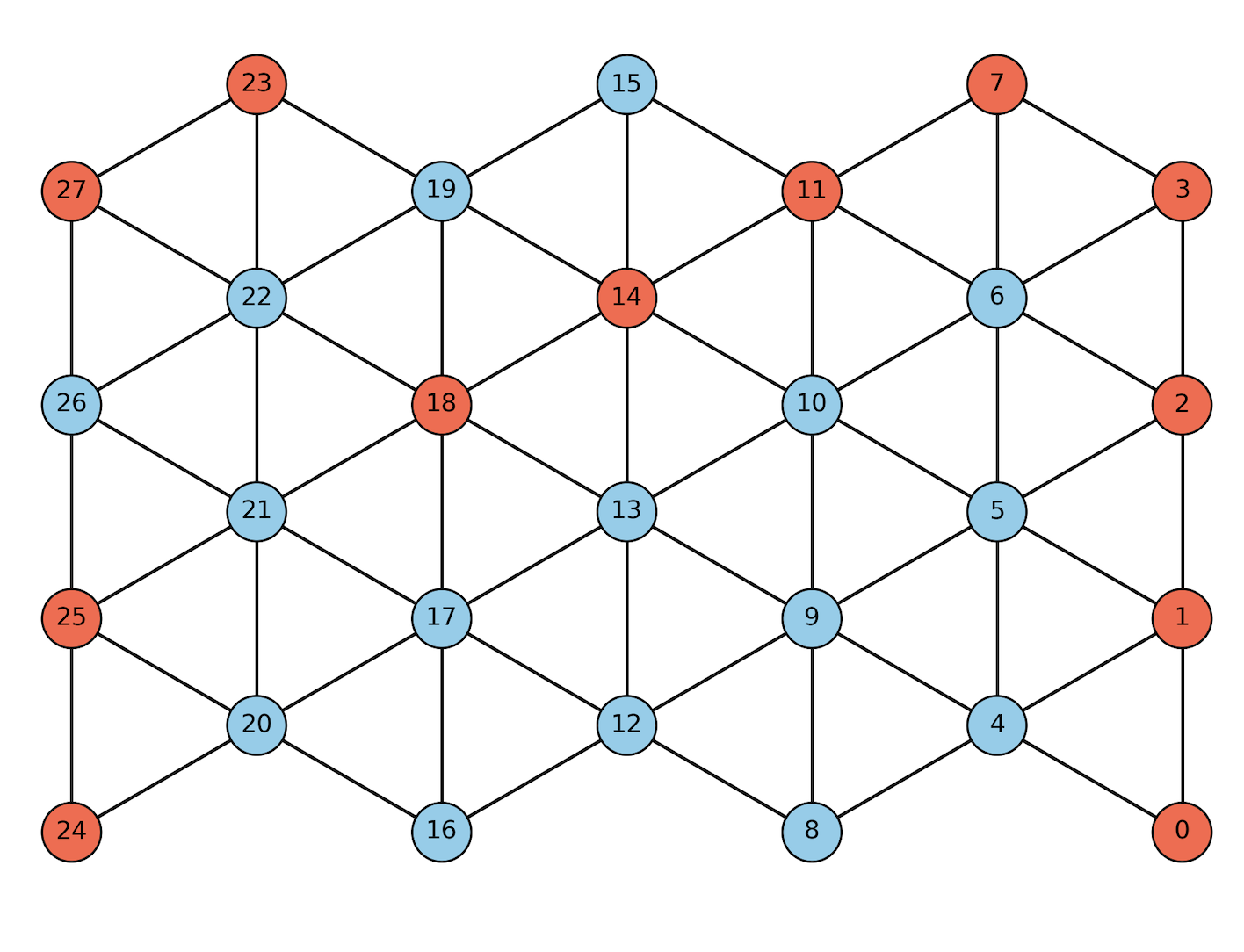}
            \caption{Rotated triangular lattice example}
        \end{subfigure}
        \caption{The red vertices indicate the vertex was randomly chosen to construct sampled root sets. This is an example of a pair of isomorphic graphs.}
        \label{fig:iso_pair}
    \end{figure}
    
    \begin{figure}
        \centering        
        \begin{subfigure}[t]{0.48\textwidth}
            \centering
            \includegraphics[height=4cm]{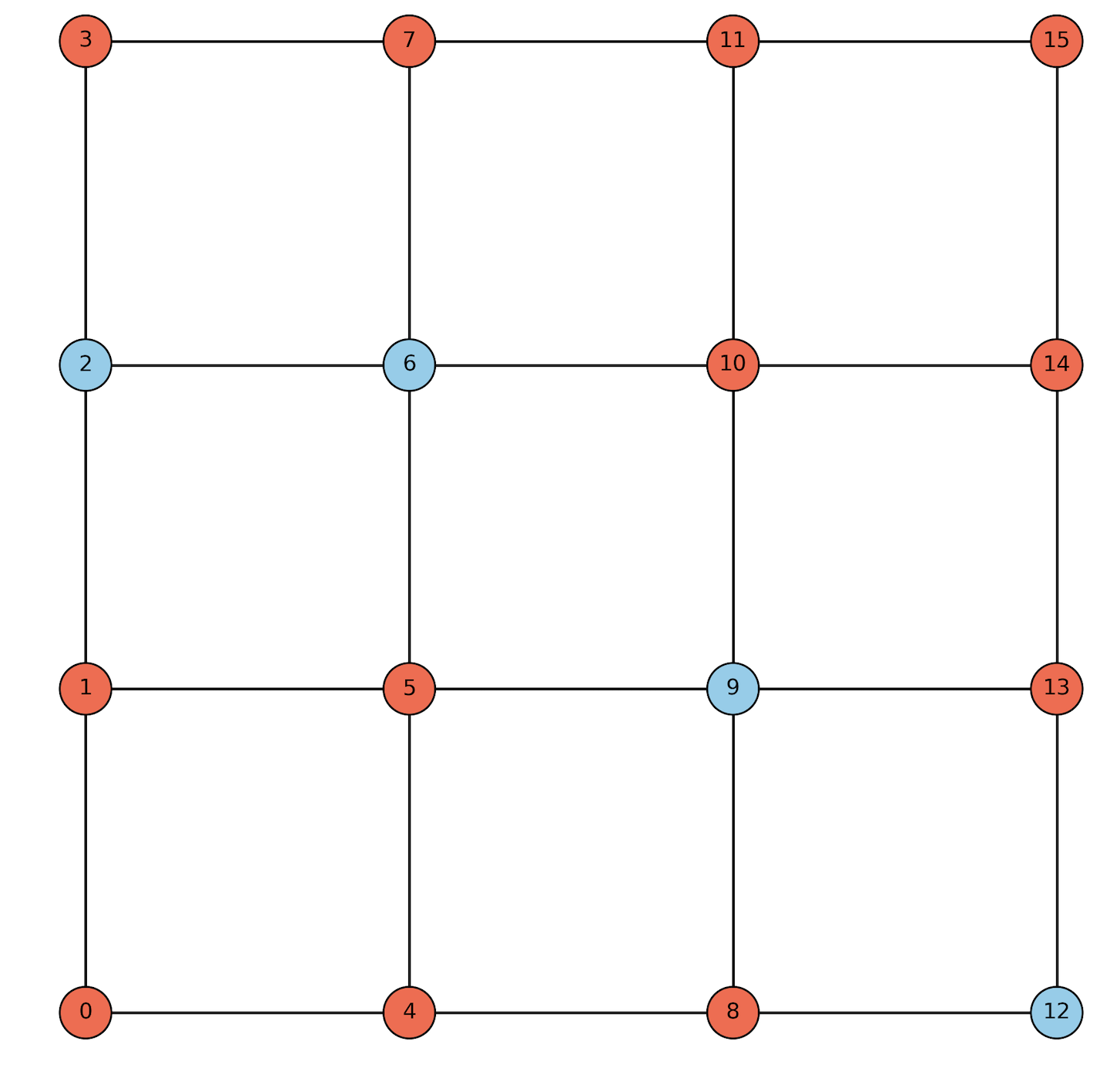}
            \caption{Square lattice example}
        \end{subfigure}
        \hfill
        \begin{subfigure}[t]{0.48\textwidth}
            \centering
            \includegraphics[height=4cm]{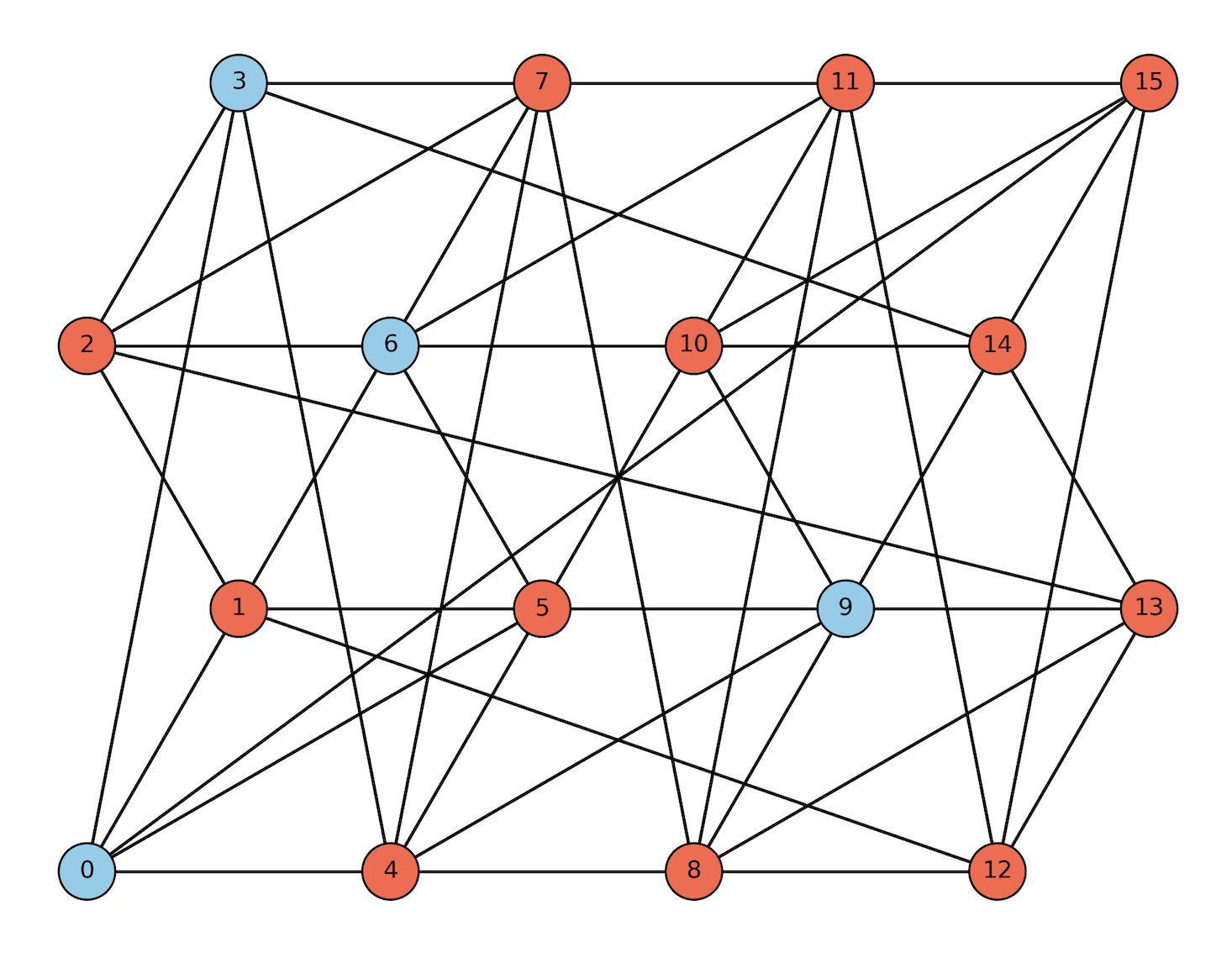}
            \caption{Triangular torus example}
        \end{subfigure}
        
        \caption{The red vertices indicate the vertex was randomly chosen to construct sampled root sets. This is an example of a pair of non-isomorphic graphs.}
        \label{fig:square_triang}
    \end{figure}

    \begin{figure}
        \centering        
        \begin{subfigure}[t]{0.48\textwidth}
            \centering
            \includegraphics[height=4cm]{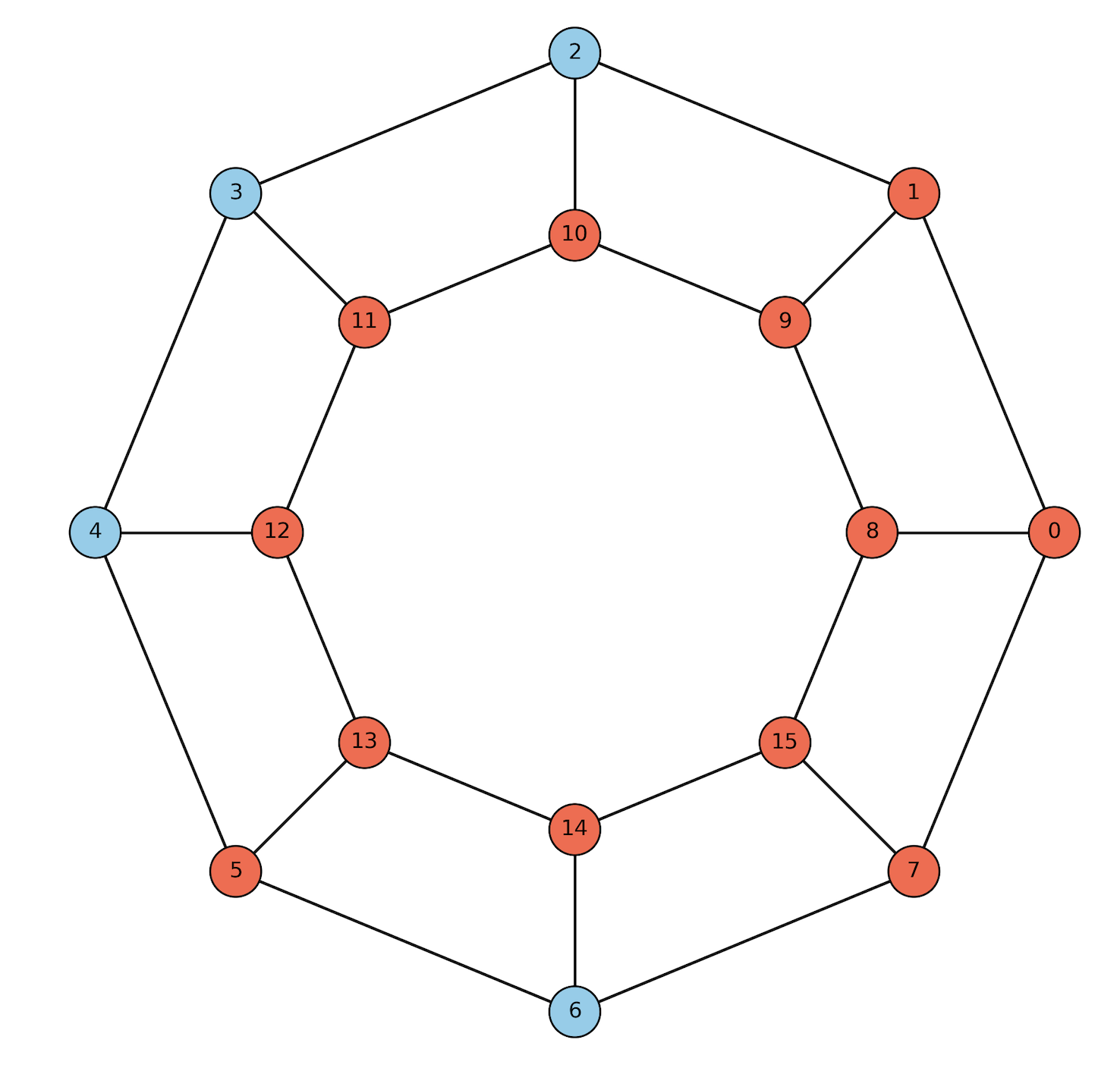}
            \caption{CL8 Circular ladder example}
        \end{subfigure}
        \hfill
        \begin{subfigure}[t]{0.48\textwidth}
            \centering
            \includegraphics[height=4cm]{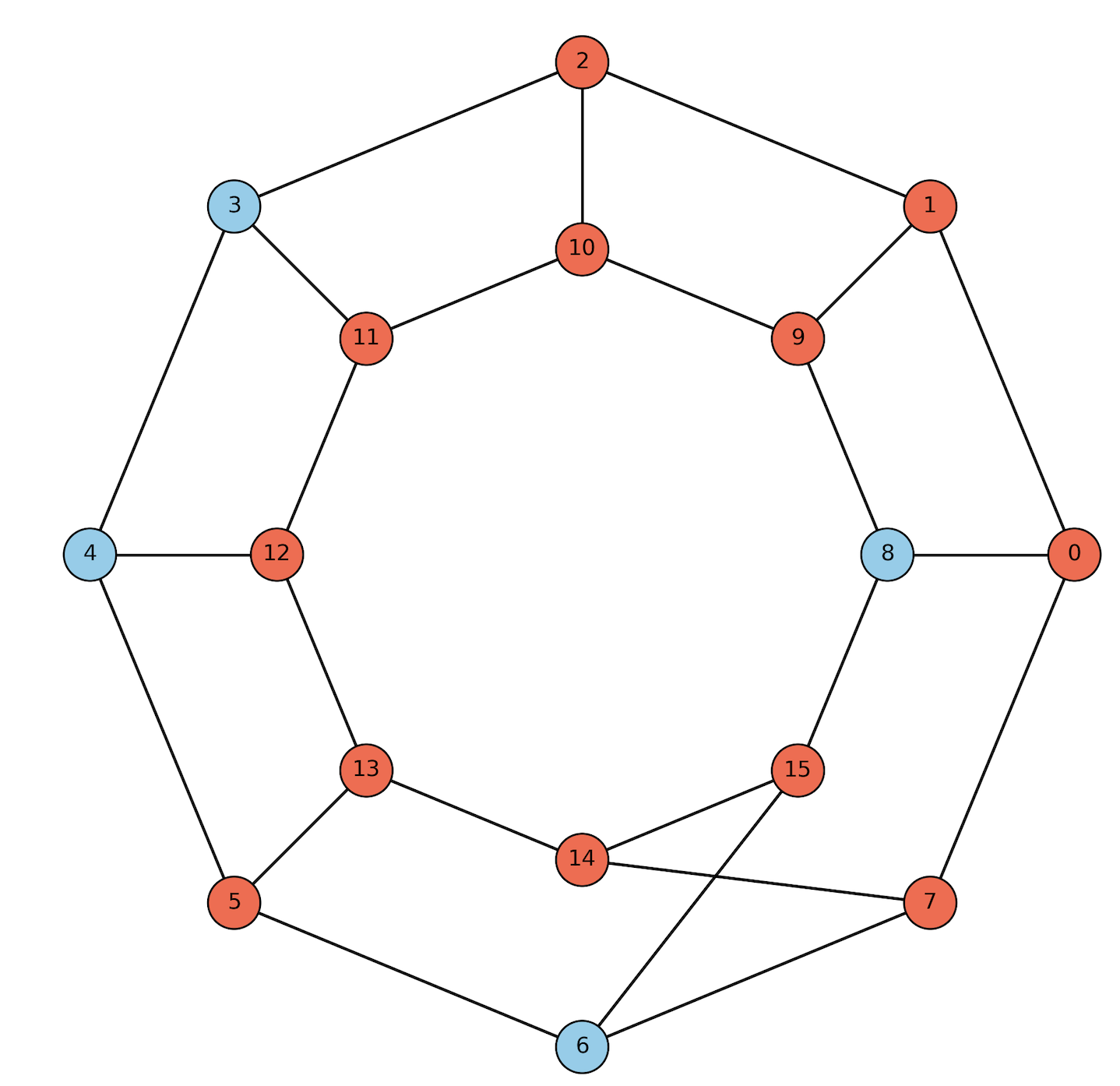}
            \caption{ML8 Mobius ladder example}
        \end{subfigure}   
        \caption{The red vertices indicate the vertex was randomly chosen to construct sampled root sets. This is an example of a pair of non-isomorphic graphs.}
        \label{fig:cl_ml}
    \end{figure}

\begin{table}
    \centering
    \renewcommand{\arraystretch}{1.2}
    \begin{tabular}{c c c c c c}
    \textbf{Graph} &
    \textbf{$h$} &
    \shortstack{\textbf{Sample}\\\textbf{Size}} &
    \shortstack{\textbf{Root Set}\\\textbf{Size}} &
    \shortstack{\textbf{Lie Algebra}\\\textbf{Dimension}} &
    \shortstack{\textbf{Commutant}\\\textbf{Dimension}} \\
    \hline
    Triangular lattice & 1 & 12 & 28 & 63 & 1 \\
    Rotated triangular lattice & 1 & 12 & 28 & 63 & 1 \\
    \hline
    Triangular lattice & 2 & 12 & 296 & 1023 & 1 \\
    Rotated triangular lattice & 2 & 12 & 296 & 1023 & 1 \\
    \hline
    Square lattice & 2 & 12 & 16 & 63 & 1 \\
    Triangular torus & 2 & 12 & 26 & 63 & 1 \\
    \hline
    CL8 circular ladder & 1 & 12 & 22 & 63 & 1 \\
    ML8 Mobius ladder & 1 & 12 & 76 & 255 & 1 \\
    \hline 
    Penrose Robinson patch & 2 & 50 & 496 & 1023 & 1 \\
    Penrose kite-dart patch & 2 & 50 & 1592 & 4095 & 1 \\
    \hline
    \end{tabular}
    \caption{Representative periodic and aperiodic graph realizations together with the computed local Pauli-support invariant data. The first and second pairs of graphs are visualized in Figure \ref{fig:iso_pair}. Figure \ref{fig:square_triang} is an example of a pair of non-isomorphic graphs, and their root data is presented here. Notice that for this example pair, the Lie information is the same however the root data is different. This demonstrates how the Lie algebra dimension is too coarse of an invariant.  Figure \ref{fig:cl_ml} is another example pair demonstrating how different graphs give rise to different root data immediately at hop size $h=1$. Finally, Figure \ref{fig:aperiodic} represents a pair of aperiodic graphs, which we discuss in the next section, that are different realizations of the same underlying Penrose quasicrystalline structure.} 
    \label{tab:comparison}
\end{table}    

It is worth mentioning, empirically, for almost every graph we tested, the Lie algebra generated by the root sets are maximal. This implies the root sets are sufficiently rich Pauli sets. Hence, the graph geometry is highly expressive, and this is the case across all our experiments, despite having very different combinatorial information. So while the Lie algebra is too coarse to distinguish graphs, the Pauli-support itself retains the distinguishing information. 

\section{Quasicrystalline (aperiodic) Structures}\label{section:aperiodic}

The periodic graphs setting benefits from the fact that lattice actions propagate local structures globally. This symmetry underlies the reconstruction results of the previous sections. In contrast, many naturally occurring structures do not admit any local translational symmetry as in \Cref{fig:aperiodic}. These aperiodic or quasicrystalline graphs lack a unit cell, and therefore fall outside the scope of lattice-based classification. This raises a fundamental question:
\begin{center}
    To what extent can local operator-theoretic information distinguish graphs in the absence of periodic structure?
\end{center}

Unlike periodic graphs, an aperiodic graph admits no global translational symmetry. However, they are still of interest because there may still be finitely many $h$-hop neighborhoods. i.e., local patterns may still recur throughout the graph in a non-periodic fashion. This suggests that invariants defined on local neighborhoods may still capture meaningful structural information. The formal definition is as follows:
\begin{definition}
    A graph $G=(V,E)$ embedded in $\mathbb{R}^d$ is called aperiodic if there does not exist a nonzero translation vector $\lambda\in\mathbb{R}^d$ such that
    \begin{equation}
        (v,w)\in E \iff (v+\lambda,w+\lambda)\in E
    \end{equation}
    for all $(v,w)\in E$.    
\end{definition}

While this removes the possibility of reducing graphs to their unit cells/finite motifs, many aperiodic graphs still exhibit finite local complexity.
\begin{definition}\label{def: flc}
    A graph $G=(V,E)$ is said to have finite local complexity (FLC) if for every hop size $h\geq 0$, there exists only finitely many rooted $h$-hop neighborhoods up to graph isomorphism. i.e. the set 
    \begin{equation}\label{eq: FLC}
        \mathcal{N}_h(G)\coloneqq \{B_h^G(v)\}_{v\in V}
    \end{equation}
    is finite. 
\end{definition}
Examples of aperiodic graphs with finite local complexity are notably the Penrose graph tilings~\cite{de1981algebraic}, Ammann--Beenker tilings~\cite{beenker1982algebraic}, and Robinson tilings~\cite{sadun2008topology}. Unlike periodic graphs, quasicrystalline graphs exhibit local patterns that repeat without any global translations.
More specifically, for a graph $G$ constructed by the Penrose tiling, using $h=1$, some vertices have degree 3, 4, 5, etc. Thus $\mathcal{N}_1(G) = \{B_1^G(v_i)\}_{i=1}^m<\infty$, for some finite $m$. 

    \begin{figure}[H]
        \centering        
        \begin{subfigure}[t]{0.48\textwidth}
            \centering
            \includegraphics[height=4cm]{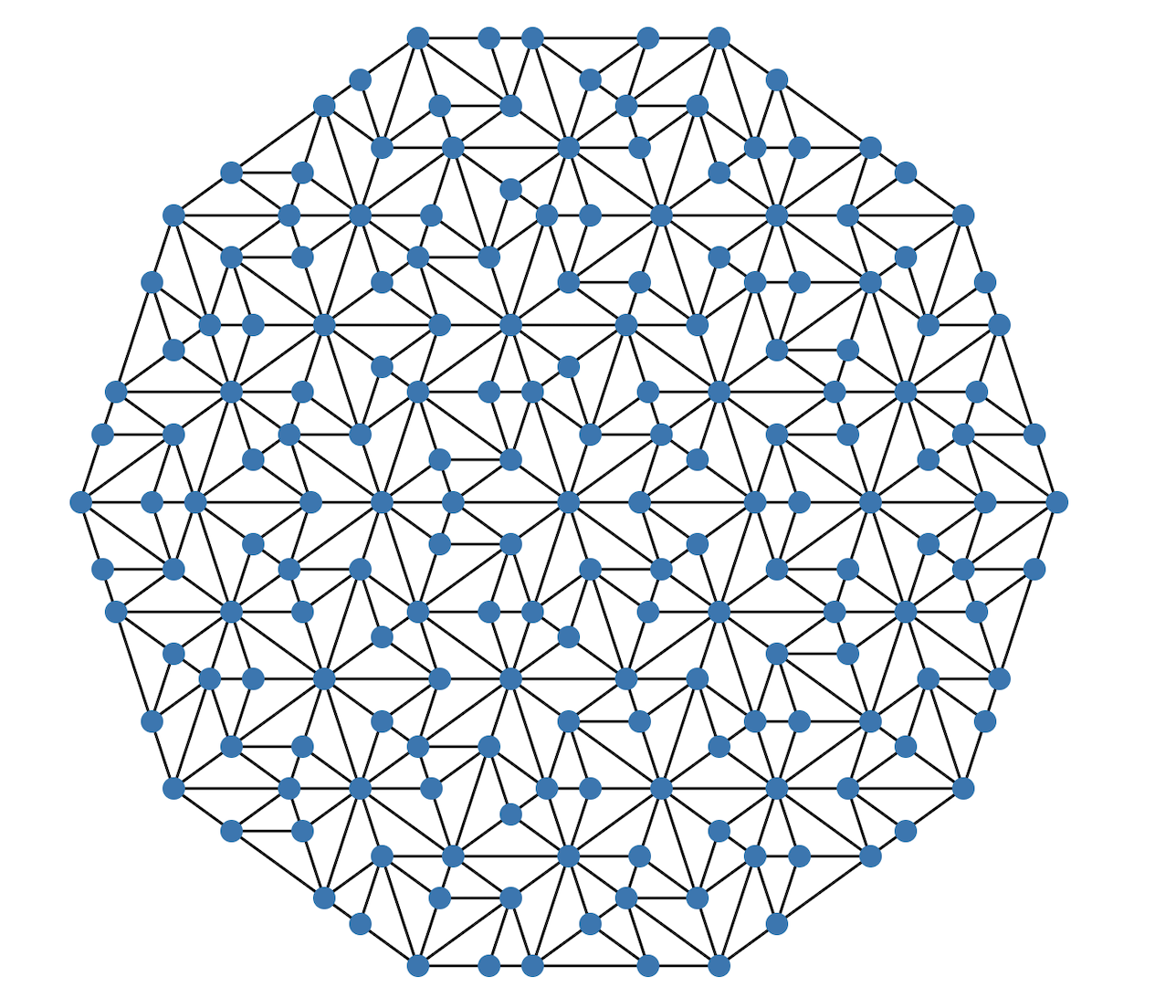}
            \caption{Penrose Robinson example}
        \end{subfigure}
        \hfill
        \begin{subfigure}[t]{0.48\textwidth}
            \centering
            \includegraphics[height=4cm]{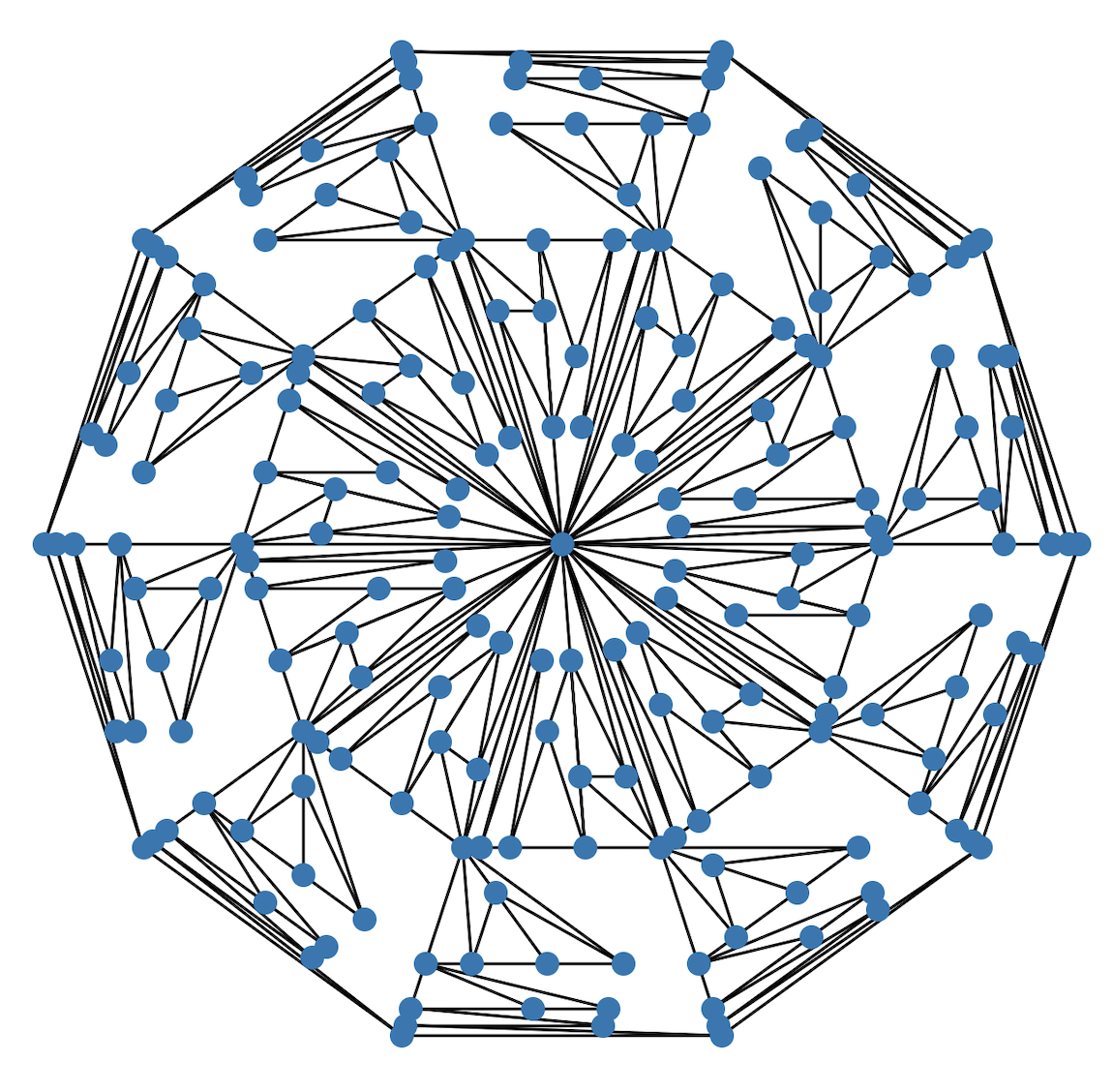}
            \caption{Penrose kite-dart example}
        \end{subfigure}        
        \caption{The coloring of the vertices was omitted to see the graphs more easily. This is an example of a pair of aperiodic graphs. } 
        \label{fig:aperiodic}
    \end{figure} 

Although there is no global translational symmetry in aperiodic graphs, local patterns recur throughout the graph in a non-periodic manner. This suggests that invariants defined on rooted local neighborhoods may still capture meaningful structural information in the absence of translational symmetry. To formalize this, we introduce rooted neighborhoods:
\begin{definition}
    Let $G=(V,E)$ be a graph and let $v\in V$. The rooted $h$-hop neighborhood of $v$ is the pair
    \begin{equation}
        (B_h^G(v),v),
    \end{equation}
    where $G[B_h^G(v)]$ denotes the induced subgraph on the vertex set
    \begin{equation}
        V(B_h(v)) = \{u\in V:\operatorname{dist}_G(u,v)\le h\}.
    \end{equation}
    Two rooted neighborhoods $(B_h(v),v)$ and $(B_h(w),w)$ are said to be rooted-isomorphic if there exists a graph isomorphism
    \begin{equation}
        \varphi:B_h(v)\to B_h(w)
    \end{equation}
    such that $\varphi(v)=w.$
\end{definition}
This definition is stronger than the standard graph isomorphism because it requires the map to send the central vertex of one graph to the central vertex of the other. The following theorem shows that the local Pauli-support framework remains
well-defined even in the absence of periodicity. This is precisely what generates a new description of quasicrystalline graphs. 

\begin{theorem}\label{theorem: flc}
    Let $G$ be a graph of finite local complexity. Then for every fixed radius $h$, only finitely many local root sets $R_h^G(v)$ for $v\in V$ can occur, i.e., the set 
    \begin{equation}
        \{R_h^G(v) : v\in V\}
    \end{equation} 
    is finite.
\end{theorem}

\begin{proof}
    Fix $h\geq 0$. Since $G$ has FLC, \Cref{eq: FLC}    is finite. 

    By construction, the local adjacency matrices are formed using the fixed canonical vertex ordering induced earlier in the paper. Consequently, if 
    \begin{equation}
        (B_h(v),v)\cong(B_h(w),w),
    \end{equation}
    then, the canonical ordering produces identical ordered adjacency matrices. Hence their skew-hermitian operators  are identical: $H_h(v)=H_h(w)$. Since the Pauli expansion of an operator is unique, the corresponding root sets coincide: $R_h(v)=R_h(w)$. Thus, the assignment
    \begin{equation}
        F:\mathcal{N}_h(G)\longrightarrow \{R_h(v):v\in V(G)\}, \qquad [(B_h(v),v)] \longmapsto R_h(v)
    \end{equation}
    is well-defined. The set $\mathcal{N}_h^G$ is finite, so is its image
    \begin{equation}
        F(\mathcal{N}_h(G)) = \{R_h(v):v\in V(G)\},
    \end{equation}
    which completes the proof. 
\end{proof}

\Cref{theorem: flc} demonstrates that the operator-theoretic framework developed in this paper remains meaningful beyond the periodic setting. Although results of the previous sections no longer apply, FLC guarantees that the collection of local root set invariants remains finite and therefore computable. 

\subsection{Cut-and-Project Lifts}
Many quasicrystalline graphs admit a cut-and-project description~\cite{de1981algebraic, beenker1982algebraic, sadun2008topology, goldman1991quasicrystal}. The ambient lattice $\Lambda$ possesses a unit cell and translation action, however the projected aperiodic graph generally does not. In this framework, an aperiodic graph in a physical space is obtained from a higher-dimensional periodic lattice. Recall the cut-and project description of aperiodic graphs. Let
\begin{equation}
    \mathbb{R}^{m+k}=E_{\parallel}\oplus E_{\perp},
\end{equation}
where $E_{\parallel}\cong \mathbb{R}^m$ is the physical space,  $E_\perp\cong \mathbb{R}^k$ is the internal space, and $E_{\parallel}$ and $E_{\perp}$ are complementary linear subspaces of
$\mathbb{R}^{m+k}$ satisfying
\begin{equation}
    E_{\parallel}\cap E_{\perp}=\{0\},
    \qquad    E_{\parallel}+E_{\perp}=\mathbb{R}^{m+k}.
\end{equation}
 Let
\begin{equation}
     \Lambda \subset \mathbb{R}^{m+k}
\end{equation}
be a lattice, and let 
\begin{equation}
    \pi_{\parallel}:\mathbb{R}^{m+k}\to E_{\parallel},
    \qquad
    \pi_{\perp}:\mathbb{R}^{m+k}\to E_{\perp}
\end{equation}
denote the corresponding projections. Given a subset $W\subset E_{\perp}$, the associated cut-and-project set is
\begin{equation}
    \mathcal{P}(W) = \left\{
    \pi_{\parallel}(x):
    x\in \Lambda,\ \pi_{\perp}(x)\in W
    \right\}.
\end{equation}
when the projection $\pi_\parallel$ is injective on the lattice points, the set $\mathcal{P}(W)$ is a quasi-periodic point set in the physical space. Graphs associated to Penrose tilings, Ammann--Beeker tilings, and related quasicrystalline structures can be described in this way~\cite{fang2021curled}. 

Although the ambient lattice $\Lambda$ is periodic in $\mathbb{R}^{m+k}$, the projected structure $\mathcal{P}(W)$ need not be periodic in $E_\parallel$. Thus, the previous results depend essentially on the existence of a unit cell and a lattice action. After projection, this global translational symmetry no longer is visible in the physical graph. 
However, cut-and-project graphs often retain finite local complexity as defined in \Cref{def: flc}. This suggests a possible strategy for extending the present periodic framework to the aperiodic setting. Our numerical experiments suggest that local root-set invariants already distinguish several aperiodic graph realizations arising from Penrose-type constructions, as demonstrated in \Cref{tab:comparison}. Establishing a rigorous connection between these invariants and the hidden periodic parent lattice remains an interesting direction for
future work.

\section{Conclusions}
In this paper, we introduced a local operator-theoretic invariant for periodic graphs by encoding rooted graph neighborhoods as adjacency Hamiltonians and recording their Pauli-support sets. We proved that these root sets are preserved under lattice-compatible graph equivalence. Conversely, under a root separation condition, we also recover the corresponding local rooted graph structure. We further show that the associated Lie algebras, Cartan decompositions, and KAK factorizations are preserved under graph relabeling. Finally, we partially extend this constructed invariant to aperiodic and quasicrystalline graphs. 

For mathematicians, this framework provides a new bridge between graph invariants, periodic geometry, and operator algebras. For physicist, this work gives an algebraic language for comparing graph-derived Hamiltonians and their dynamical structures, which can be simulated on quantum computers. For material scientists, it offers a local method for distinguishing and analyzing crystalline and quasicrystalline patterns. Future work may extend the framework to broader classes of graph structures, weighted or directed graph classes, and more physically detailed Hamiltonians.

\backmatter

\section*{Declarations}

\begin{itemize}
\item \textbf{Funding} 
S.C. and R.M. were supported by NIST through the CIPP program under Award No. 60NANB24D218 and NSF through the NSF TIP program under Award No. 2534232.
\item \textbf{Competing interests statement:} The authors have no competing interests to declare that are relevant to the content of this article.
\item \textbf{Ethics approval and consent to participate:} not applicable
\item \textbf{Consent for publication:} not applicable
\item \textbf{Data availability:} not applicable 
\item \textbf{Materials availability:} not applicable
\item \textbf{Code availability: } Upon request 
\end{itemize}

\bibliography{bib}
\end{document}